%% file: PRLv2.tex
\documentclass[prl,superscriptaddress,nopacs,twocolumn]{revtex4-1}

\usepackage{xcolor}
\definecolor{darkred}{rgb}{0.8,0.1,0.1}

\usepackage{amsmath,mathtools}
\usepackage{enumerate}
\usepackage{graphicx,epsfig,amsmath,amssymb,verbatim,color}
\usepackage{dsfont}
\usepackage{float}
\usepackage{tikz}

\usepackage[T1]{fontenc}
\usepackage[utf8]{inputenc}

\newcommand{\SIO}{\text{\rm SIO}}
\newcommand{\IO}{\text{\rm IO}}
\newcommand{\DIO}{\text{\rm DIO}}
\newcommand{\MIO}{\text{\rm MIO}}
\newcommand{\SSIO}{\text{\rm (S)IO}}
\newcommand{\Choi}{Choi-Jamio\l{}kowski }
\input{pretex}

\nc{\subO}{\cO_{\text{\rm sub}}}
\nc{\subMIO}{\MIO_{\text{\rm sub}}}

\usepackage{pdfpages}
\makeatletter
\AtBeginDocument{\let\LS@rot\@undefined}
\makeatother

\nc{\uts}{\affiliation{Centre for Quantum Software and Information, Faculty of Engineering and Information Technology, \\ University of Technology Sydney, NSW 2007, Australia}}
\nc{\notts}{\affiliation{School of Mathematical Sciences and Centre for the Mathematics and Theoretical Physics of Quantum Non-Equilibrium Systems,\\University of Nottingham, University Park, Nottingham NG7 2RD, United Kingdom}}

\begin{document}
	\title{Probabilistic distillation of quantum coherence}
	\author{Kun Fang}
	\email{kun.fang-1@student.uts.edu.au}
	\uts
	
	\author{Xin Wang}
	\email{xin.wang-8@student.uts.edu.au}
	\uts
	
	\author{Ludovico Lami}
	\email{ludovico.lami@gmail.com}
	\notts
	
	\author{Bartosz Regula}
	\email{bartosz.regula@gmail.com}
	\notts
	
	\author{Gerardo Adesso}
	\email{gerardo.adesso@nottingham.ac.uk}
	\notts
	

\begin{abstract}
	
The ability to distill quantum coherence is pivotal for optimizing the performance of quantum technologies; however, such a task cannot always be accomplished with certainty. Here we develop a general framework of {\em probabilistic} distillation of quantum coherence in a one-shot setting, establishing fundamental limitations for different classes of free operations.
We first provide a geometric interpretation for the maximal success probability, showing that under maximally incoherent operations (MIO) and dephasing-covariant incoherent operations (DIO) the problem can be simplified into efficiently computable semidefinite programs. Exploiting these results, we find that DIO and its subset of strictly incoherent operations (SIO) have equal power in probabilistic distillation of coherence from pure input states, while MIO are strictly stronger. We then prove a fundamental no-go result: distilling coherence from any full-rank state is impossible even probabilistically. We further find that in some conditions the maximal success probability can vanish suddenly beyond a certain threshold in the distillation fidelity. Finally, we consider probabilistic coherence distillation assisted by a catalyst and demonstrate, with specific examples, its superiority to the unassisted and deterministic cases.
\end{abstract}

\maketitle

Quantum coherence is a physical resource that is essential for various tasks in quantum computing (e.g.~Deutsch-Jozsa algorithm~\cite{hillery2016coherence}), cryptography (e.g.~quantum key distribution~\cite{coles2016numerical}), information processing (e.g.~quantum state merging~\cite{streltsov_2016}, state redistribution~\cite{anshu2018quantum} and channel simulation~\cite{diaz2018}), thermodynamics~\cite{lostaglio2015description}, metrology~\cite{Frowis2011}, and quantum biology~\cite{streltsov_2017}. A series of efforts have been devoted to building a resource framework of coherence in recent years~\cite{aberg_2006,gour_2008,levi_2014,baumgratz_2014,streltsov_2017}, characterizing in particular the state transformations and operational uses of coherence in fundamental resource manipulation protocols~\cite{winter_2016,chitambar_2016-3,streltsov_2016,chitambar_2016-2,zhao_2018,regula_2017-1}. As in any physical resource theory, a central problem of the resource theory of quantum coherence is \emph{distillation}: the process of extracting canonical units of coherence, as represented by the maximally coherent state $\ket{\Psi_m}$, from a given quantum state using a choice of free operations.

The usual asymptotic approach to studying the problem in quantum information theory is to assume that there is an unbounded number of independent and identically distributed copies of a quantum state available and that the transformation error asymptotically goes to zero~\cite{bennett_1996-1,bennett_1996-3,rains_1999,winter_2016}. In reality, these assumptions become unphysical due to our limited access to a finite number of copies of a given state, making it necessary to look at non-asymptotic regimes \cite{zhao_2018,regula_2017-1}. Furthermore, since loss and decoherence severely restrict our ability to manipulate large quantum systems, one needs to allow for a finite error in the distillation protocol. In this respect, deterministic protocols such as those studied in \cite{regula_2017-1} may be insufficient to reach a target  fidelity for desired applications. It is thus of importance to consider a more general {\em probabilistic} framework, in which the distillation will succeed only with some probability. Here, the allowed error can be characterized by two key parameters: the success probability of the one-shot distillation process, and the fidelity between the extracted state and the target state $\ket{\Psi_m}$. To have a systematic understanding of coherence distillation with finite resources and be able to implement practical schemes for this task, it is crucial to describe and optimize the interplay between these two parameters.

In this Letter, we develop the framework of probabilistic coherence distillation, characterizing the relation between the maximum success probability and the fidelity of distillation in the one-shot setting. We describe qualitative and quantitative aspects of this task under several representative choices of free operations, providing insights into their fundamental limitations and capabilities for coherence manipulation. In particular, we achieve a complete characterization of probabilistic coherence distillation with pure input states. The main results of our study are presented as Theorems in the following, with all proofs delegated to the Supplemental Material \footnote{See the Supplemental Material for technical derivations and proofs of the results in the Letter, which includes Refs.~[22-25] }. \nocite{zhang_2005,bhatia_2007,ben-israel_2003,chitambar_2016,regula_2017-1}
Before proceeding, we note that, previously, the framework of probabilistic state transformations has been employed in characterizing entanglement distillation~\cite{bennett_1996-1,lo_2001,vidal_1999-1,jonathan_1999,ishizaka_2005} as well as related settings in the resource theory of thermodynamics~\cite{alhambra_2016}, and recently found use in the investigation of practical entanglement distillation schemes~\cite{rozpedek_2018}. Our work fills an important gap in the literature by establishing the probabilistic toolbox for the key resource of quantum coherence.

\vspace{0.1cm}
\textbf{\textit{Framework of probabilistic coherence distillation.}}---The free states in the resource theory of quantum coherence, so-called incoherent states $\cI$, are the density operators which are diagonal in a given reference orthonormal basis $\{\ket{i}\}$. We will use $\Delta(\cdot):=\sum_i \ket{i}\bra{i} \cdot \ket{i}\bra{i}$ to denote the diagonal map (completely dephasing channel) in this basis. The resource theory of coherence is known not to admit a unique physically-motivated choice of allowed free operations \cite{winter_2016,chitambar_2016,marvian_2016,vicente_2017,streltsov_2017}, necessitating the investigation of operational capabilities of several different classes of maps. The relevant choices of free operations that we will focus on are: \emph{maximally incoherent operations (MIO)}~\cite{aberg_2006}, defined to be all operations $\cE$ such that $\cE(\rho) \in \cI$ for every $\rho \in \cI$;  \emph{dephasing-covariant incoherent operations (DIO)}~\cite{chitambar_2016,marvian_2016}, which are maps $\cE$ such that $[\Delta, \cE] = 0$, or equivalently $\cE(\ket{i}\bra{i}) \in \cI$ and $\Delta(\cE(\ket{i}\bra{j})) = 0, \forall i \neq j$; \emph{incoherent operations (IO)}~\cite{baumgratz_2014}, which admit a set of incoherent Kraus operators $\{K_l\}$ such that $\frac{K_l \rho K_l^\dagger}{\tr K_l \rho K_l^\dagger} \in \cI$ for all $l$ and $\rho\in \cI$; as well as \emph{strictly incoherent operations (SIO)}~\cite{winter_2016}, which are operations such that both $\{K_l\}$ and $\{K_l^\dagger\}$ are sets of incoherent operators. In particular, MIO is the largest possible choice of free operations in coherence theory,
while SIO can be regarded as the smallest class which satisfies desirable resource theoretic criteria \cite{chitambar_2016,streltsov_2017}, leading to the hierarchy $\SIO \subsetneq \IO \subsetneq \MIO$, $\SIO \subsetneq \DIO \subsetneq \MIO$ \cite{chitambar_2016}.

The basic task of probabilistic distillation can be understood as follows. For any given quantum state $\rho$ held by a single party $A$, we aim to transform this state to an $m$-dimensional maximally coherent state (target state) $\ket{\Psi_m}:=m^{-1/2}\sum_{i=1}^m \ket{i}$ with high fidelity. A single-bit classical flag register $L$ is used to indicate whether the transformation succeeds or not. If the flag is found in the $0$ state, the distillation process has succeeded and the output state $\sigma$ has fidelity at least $1-\ve$ with the target state. Otherwise, the process has failed, and we discard the unwanted output state $\omega$.
Our goal is to maximize the success probability while keeping the transformation infidelity within some tolerance $\ve$.

Formally, for any triplet $(\rho,m,\ve)$ with a given initial state $\rho$, target state dimension $m$, and infidelity tolerance $\ve$, the maximal success probability of coherence distillation under the operation class $\cO \in \{\SIO,\IO,\DIO,\MIO\}$ is denoted as $P_{\cO}(\rho\!\to\! \Psi_m, \ve)$, where  $\Psi_m:=\ket{\Psi_m}\bra{\Psi_m}$. This is given by the maximal value of $p$ such that there exists a transformation $\Pi_{A\to LB}\in \cO$ satisfying the constraints
\begin{equation}\label{prob 1}
\begin{aligned}
  \Pi_{A\to LB}(\rho) &= p \ket{0}\bra{0}_L \ox \sigma + (1-p) \ket{1}\bra{1}_L\ox \o,\\
  F(\sigma,\Psi_m) &\geq 1-\ve,
\end{aligned}
\end{equation}
where $F(\rho,\sigma):=\|\sqrt{\rho}\sqrt{\sigma}\|_1^2$ is the fidelity and $\|\cdot\|_1$ is the trace norm.
If the distillation fails, we can perform a free operation to make the unwanted state $\omega$ completely mixed without changing the success probability. Thus, without loss of generality, we can take $\omega = \1/m$.
Exploiting the fact that the target state $\Psi_m$ is invariant under the twirling operation $\cT(\rho) = \frac1{d!} \sum_{i=1}^{d!} P_i \rho P_i$, where $P_i$ are all the permutations on the input system and $d$ is the input dimension, we can also fix the optimal output state as $\sigma = \Psi_m^\ve$ where $\Psi_m^\ve:=(1-\ve)\Psi_m+\ve (\1-\Psi_m)/(m-1)$.
Specifically, for any optimal output state $\sigma$, we can further perform the free operation $\cT$, which gives a new output state $\cT(\sigma)$ always in the form of $ a \Psi_m + b (\1 - \Psi_m)/(m-1)$,
where we can choose $a = 1-\ve$ and $b = \ve$ while keeping the  fidelity with the target state and the optimal success probability unchanged.
This allows us to write
 $P_{\cO}(\rho\!\to\!\Psi_m,\ve) = P_{\cO}(\rho\!\to\!\Psi_m^\ve,0)$,
meaning that the maximal success probability of coherence distillation is the same as the maximal success probability of transforming the given state to the target  $\Psi_m^\ve$ with fidelity one.

\begin{figure}[t]
\centering
\includegraphics[width=3.8cm]{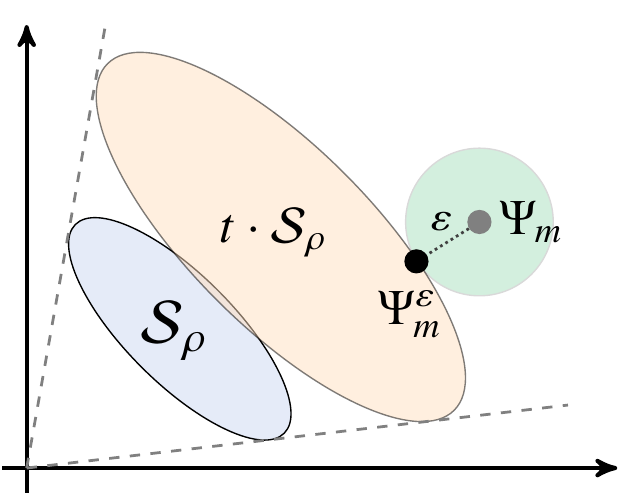}
\caption{Geometric interpretation of the maximal success probability of coherence distillation. See text for details.}
\label{geo}
\end{figure}

\vspace{0.1cm}
\textbf{\textit{Computing the maximum distillation probability.}}---
We now set out to find efficiently computable expressions for the maximal distillation probability.
Consider a generalization of the set $\cO$ to the class $\subO$ of subnormalized quantum operations, that is, completely positive and trace-nonincreasing maps.
Using this notation, we can conveniently express the maximal success probability as follows (see also \cite{ishizaka_2005,buscemi_2017}).
\vspace{-.25\baselineskip}
\begin{proposition}\label{proposition 1}
  For any triplet $(\rho,m,\ve)$ and operation class $\cO$, the maximal success probability $P_{\cO}(\rho\!\to\! \Psi_m,\ve)$ is given by
    $\max \lsetr p\in \mathbb{R_+} \barr \cE(\rho) = p \cdot\Psi_m^\ve,\ \cE \in \subO \rsetr$.
It then holds that $P_{\cO}\left(\rho\!\to\! \Psi_m,\ve\right)^{-1} = \min \big\{ t\in \mathbb{R_+} \big| \Psi_m^\ve \in t \cdot \cS_\rho \big\}$ where $\cS_\rho := \big\{ \cE(\rho) \,\big|\, \cE \in \subO \big\}$ is the set of all the output operators of $\rho$ under the operation class $\subO$.
\end{proposition}
\vspace{-.25\baselineskip}
The result simplifies the optimization of the maximal success probability via subnormalized free operations, providing a geometric interpretation for the maximal success probability as a gauge function~\cite{rockafellar_1970}, as shown in Fig.~\ref{geo}.
This justifies our intuition that the closer the state $\rho$ to $\Psi_m^\ve$, the less we need to expand the set $\cS_\rho$, and thus the larger success probability we can obtain. We note that the result in Prop.~\ref{proposition 1} can be extended also to more general convex resource theories \cite{brandao_2015,regula_2018,chitambar_2018-1}.

By further exploiting the symmetry of $\Psi_m^\ve$, we can compute the maximal success probability under MIO/DIO via the following semidefinite programs (SDPs).
\begin{theorem}\label{SDP thm}
  For any triplet $(\rho,m,\ve)$, the maximal success probability of distillation under MIO and DIO are
  \begin{subequations}
	\begin{align}
P_{\MIO}(\rho\!\to\! \Psi_m,\ve) =	\text{\rm max.} &\ \tr G\rho \notag\\
	\text{\rm s.t.} &\ \ \Delta(G) = m \Delta(C), \label{eq:c1}\\
	&\ \ 0 \leq C \leq G \leq \1, \label{eq:c2}\\
	&\ \tr C\rho \geq (1-\ve) \tr G \rho, \label{eq:c3} \\
P_{\DIO}(\rho\!\to\! \Psi_m,\ve) = \text{\rm max.} &\ \tr G\rho \notag\\
	\text{\rm s.t.} &\ \ \text{\rm Eqs.}~(\ref{eq:c1},\ref{eq:c2},\ref{eq:c3}),\ G = \Delta(G).  \notag
\end{align}	
\end{subequations}
\end{theorem}
For completeness, we give the dual forms and alternative formulations of the SDPs in the Supplemental Material~\cite{Note1}.
These SDPs provide us with efficient ways to numerically calculate the maximal success probability for general triplets $(\rho,m,\ve)$, and allow us to obtain key results on the power of different  operations for probabilistic coherence distillation.

In this respect, let us also consider  the choice of IO or SIO as the free operations. It is known that these two sets of operations have the same power in  pure-state transformations, completely characterized by majorization relations. This yields~\cite{lo_2001,vidal_1999-1,chitambar_2016,zhu_2017-1,du2015conditions}
\begin{align*}
    P_\SSIO(\varphi\!\to\! \Psi_m,0) \!=\! \begin{cases} 0  &  \text{if } \rank\Delta(\varphi) < m,\\
    \begin{displaystyle}\min_{k\in[1,m]} \frac{m}{k} \hspace{-.3em}\sum_{i=m\!-\!k\!+\!1}^d \hspace{-.3em}\varphi_i^2 \end{displaystyle} &\text{otherwise,}\end{cases}\raisetag{2.2\baselineskip}
\end{align*}
where we have assumed without loss of generality that the coefficients of $\varphi$ are non-negative and arranged in non-increasing order.
In a similar way, operations in the class DIO can never increase the diagonal rank of a pure state, while it is known that MIO allow for the rank to increase~\cite{chitambar_2016}, suggesting that MIO is a much stronger class. It is thus surprising that MIO and DIO have exactly the same power in the task of deterministic coherence distillation~\cite{regula_2017-1}, and that the two sets of operations lead to the same asymptotic transformation rates for all states~\cite{chitambar_2017-1}.
In the following, we will instead show crucial differences between MIO and DIO when one goes beyond deterministic transformations,
highlighting the increased capabilities of MIO in probabilistic distillation, as well as establishing limitations on coherence distillation in general.
\begin{theorem}\label{full rank thm}
For any triplet $(\rho,m,0)$ with a full-rank state $\rho$, it holds that $P_{\MIO}(\rho\!\to\! \Psi_m,0) = 0$. For any triplet $(\varphi,m,0)$ with a coherent pure state $\ket{\varphi} = \sum_{i=1}^n \varphi_i\ket{i}$, $\varphi_i \neq 0$, $n \geq 2$, it holds
\begin{align}\label{MIO pure zero error}
P_{\MIO}(\varphi\!\to\! \Psi_m,0) \geq \frac{n^2}{m(\sum_{i=1}^n |\varphi_i|^{-2})} > 0.
\end{align}
\end{theorem}
This result establishes a no-go theorem for coherence distillation, showing that no class of free operations preserving incoherent states can allow to distill any perfect coherence from a full-rank state, even probabilistically. Note that any generic density matrix has full rank, and so does $\Psi_m^\ve$ for any $\ve > 0$.
Thus $ |P_{\MIO}(\Psi_m^\ve\!\to\!\Psi_m,0) - P_{\MIO}(\Psi_m\!\to\!\Psi_m,0)| = 1$,
even though $\Psi_m^\ve$ can be arbitrarily close to $\Psi_m$, implying that the maximal success probability is \textit{not continuous} with respect to the input state.
The physical implications of this result are that any noise typically resulting in full-rank states will lead in practice to an irretrievable loss of resources. For example, in a scenario where the coherent state $\Psi_m$ is stored in a quantum memory exposed to depolarizing noise, it is impossible to recover it perfectly via free operations, even probabilistically.


However, for any {\em pure} coherent state, Theorem~\ref{full rank thm} shows that it is always possible to probabilistically distill a maximally coherent state of arbitrary dimension via MIO. In the Supplemental Material~\cite{Note1} we establish a tighter bound for the probability of distillation under MIO, which in particular gives $P_{\text{MIO}}{(\Psi_{n} \!\to\! \Psi_m, 0)} \geq \frac{n-1}{m-1}$ when $m>n$.
Observe that instead $P_{\text{DIO}}(\Psi_{n} \!\to\! \Psi_m, 0)=0$ for $m>n$. This tells us that, as the dimension $n$ increases, there are $n$-dimensional density matrices $\rho_{n}$ such that $P_{\text{MIO}}(\rho_{n}\!\to\! \Psi_{n+1}, 0)\rightarrow 1$ while $P_{\text{DIO}}(\rho_{n}\!\to\! \Psi_{n+1}, 0) = 0$ for all $n$, i.e., $P_{\text{MIO}}$ and $P_{\text{DIO}}$ can exhibit an arbitrarily large gap. This shows that in the probabilistic distillation scenario, MIO can be much more powerful than DIO in general, in a stark contrast with the case of deterministic coherence distillation~\cite{regula_2017-1} (see Table~\ref{tablediagram}).

\begin{table}[ht]
{\footnotesize{
\begin{tabular}{llll}
\hline \hline
$\begin{array}{l}\mbox{Deterministic}\\ \mbox{distillation \cite{regula_2017-1}} \end{array}$ & $\begin{array}{ll} \mbox{Pure states} & \mbox{General states} \\
                             \mbox{MIO $=$ DIO $=$ IO $=$ SIO} & \mbox{MIO $=$ DIO}\end{array}$ \\ \hline
$\begin{array}{l}\mbox{Probabilistic}\\ \mbox{distillation [$\star$]} \end{array}$ & $\begin{array}{ll} \mbox{Pure states} & \mbox{Full-rank states} \\
                             \mbox{MIO $>$ DIO $=$ IO $=$ SIO} & \mbox{MIO $=$ DIO $=$ IO $=$ SIO $= \varnothing$}\end{array}$ \\
\hline \hline
\end{tabular}}}
\caption{\label{tablediagram}Comparison of the operational power of different sets of free operations for deterministic \cite{regula_2017-1} versus probabilistic [$\star$] (this paper) distillation of quantum coherence. $\varnothing$ denotes the empty set.}
\end{table}

The relation between the capabilities of different operations is made precise by the following result, characterizing the fundamental task of distilling coherence from pure input states.
\begin{theorem}\label{thm:dio sio}
For any pure state $\varphi$ and any $m$, we have
  $P_\DIO(\varphi \!\to\! \Psi_m, 0) = P_\SSIO(\varphi \!\to\! \Psi_m, 0)$.
\end{theorem}
We can therefore see that DIO does not provide any operational advantage over SIO and IO in pure-state probabilistic coherence distillation, despite being a strictly larger class than SIO~\cite{chitambar_2016,bu_2017-1}. Putting together the results of Theorems \ref{full rank thm} and \ref{thm:dio sio}, we have shown that a large gap in the operational capabilities of operations in the one-shot resource theory of quantum coherence exists between MIO and DIO, but not between DIO and SIO/IO --- this can be compared with the case of deterministic distillation, where all sets of operations $\cO \in \{\SIO,\IO,\DIO,\MIO\}$ allow for the same achievable rate of distillation from pure states~\cite{regula_2017-1}, see Table~\ref{tablediagram}.

In the task of distilling maximally coherent qubit states ${\Psi_2}$, we can extend the above result and obtain analytically the maximal success probability
for arbitrary infidelity $\ve$. In this particular case, MIO provides no advantage over DIO.
\begin{proposition}\label{pure to psi 2}
For $\cO\in \{\text{DIO}, \text{MIO}\}$ and any pure state $\varphi$ with $\varphi_1 \geq ... \geq \varphi_n > 0$, it holds that
\begin{equation*}\label{DIO pure anal}
  P_{\cO}(\varphi\!\to\!\Psi_2,\ve) \!=\! \left\{ \!\begin{array}{ll} 1 & \text{if $\ve\!\geq\! \ve_0(\varphi_1)$,} \\  2(1-\varphi_1^2) \left(\frac{\sqrt{1-\ve}+\sqrt{\ve}}{1-2\ve}\right)^2 & \text{otherwise.} \end{array} \right.
\end{equation*}
\end{proposition}
Here the function $\ve_0$, defined as $\ve_0(\varphi_1)=0$ if  $\varphi_1\leq\frac{1}{\sqrt2}$ and $\ve_0(\varphi_1)= \frac12-\varphi_1\sqrt{1-\varphi_1^2}$ otherwise, can be related to the $m$-distillation norm~\cite{regula_2017-1} characterizing the fidelity of deterministic distillation.
Using this analytical result, we can give a concrete example to show that the probabilistic distillation framework can outperform the deterministic one. Suppose we need to distill a maximally coherent qubit state $\Psi_2$ from the input state $\ket{\varphi} = (3\ket{0}+\ket{1})/\sqrt{10}$ with acceptable fidelity at least $0.9$. The input state becomes useless in the deterministic scenario, since the maximal fidelity achievable via deterministic protocols is given by $0.8$. However, probabilistic operations allow us to achieve the required distillation fidelity with $1/2$ success probability, demonstrating an explicit operational advantage of probabilistic distillation. In another scenario, if the acceptable fidelity is $0.8$, we can gain higher distillation fidelity by compromising some success probability even though deterministic protocols are sufficient to accomplish the task. Such a setting may be dubbed ``gambling with coherence'', adapting terminology from \cite{bennett_1996-1,lo_2001}.


\vspace{0.1cm}
\textbf{\textit{Relation between distillation fidelity and probability.}}---For any given input state $\rho$ and target state dimension $m$, the maximal success probability is only dependent on the transformation fidelity. The higher the fidelity we require, the lower the probability that we will succeed. Intuitively, one would expect the success probability to decrease smoothly as the fidelity increases; however, we 
will now show that the success probability can vanish discontinuously beyond some fidelity threshold. This phenomenon is analogous to the strong converse theorem in channel coding theory~\cite{wolfowitz_1978,ogawa_1999,winter_1999}, which says that the coding success probability goes to zero if the coding rate exceeds the capacity of the channel.
While this phenomenon cannot occur in distillation from pure input states under MIO due to Theorem~\ref{full rank thm}, in the following result we completely characterize this ``sudden death'' property  for pure input states under DIO.
\begin{proposition}\label{nontradeoff prop}
	For any pure state $\ket{\varphi} = \sum_{i=1}^n \varphi_i \ket{i}$ with nonzero coefficients $\varphi_i$, it holds that
	\begin{equation*}\begin{aligned}
	P_{\DIO}(\varphi\!\to\! \Psi_m,\ve) \!\ \begin{cases} >0 \;\ \ \text{ if } n \geq m \text{ or if } n < m \text{ and } \ve \geq 1 \!-\! \frac{n}{m}, \\  = 0 \;\ \ \text{ if } n < m \text{ and } \ve < 1 \!-\! \frac{n}{m}.\end{cases}
	\end{aligned}\end{equation*}
\end{proposition}

In the particular case of the transformation $\Psi_n\!\to\! \Psi_m$ with $n \leq m$, the probability equals 1 as long as $\ve \geq 1 - \frac{n}{m}$. The result shows in particular that, if the output dimension is larger than the input dimension, any trade-off between the maximal success probability and the transformation fidelity will always be truncated at the fidelity threshold $\frac{n}{m}$. Specifically, at the point $\ve = 1 - \frac{n}{m}$, demanding a slightly higher fidelity will make the probabilistic distillation impossible, as shown in Fig.~\ref{nontradeoff ex}.

\begin{figure}[t]
\centering
\includegraphics[width=5cm]{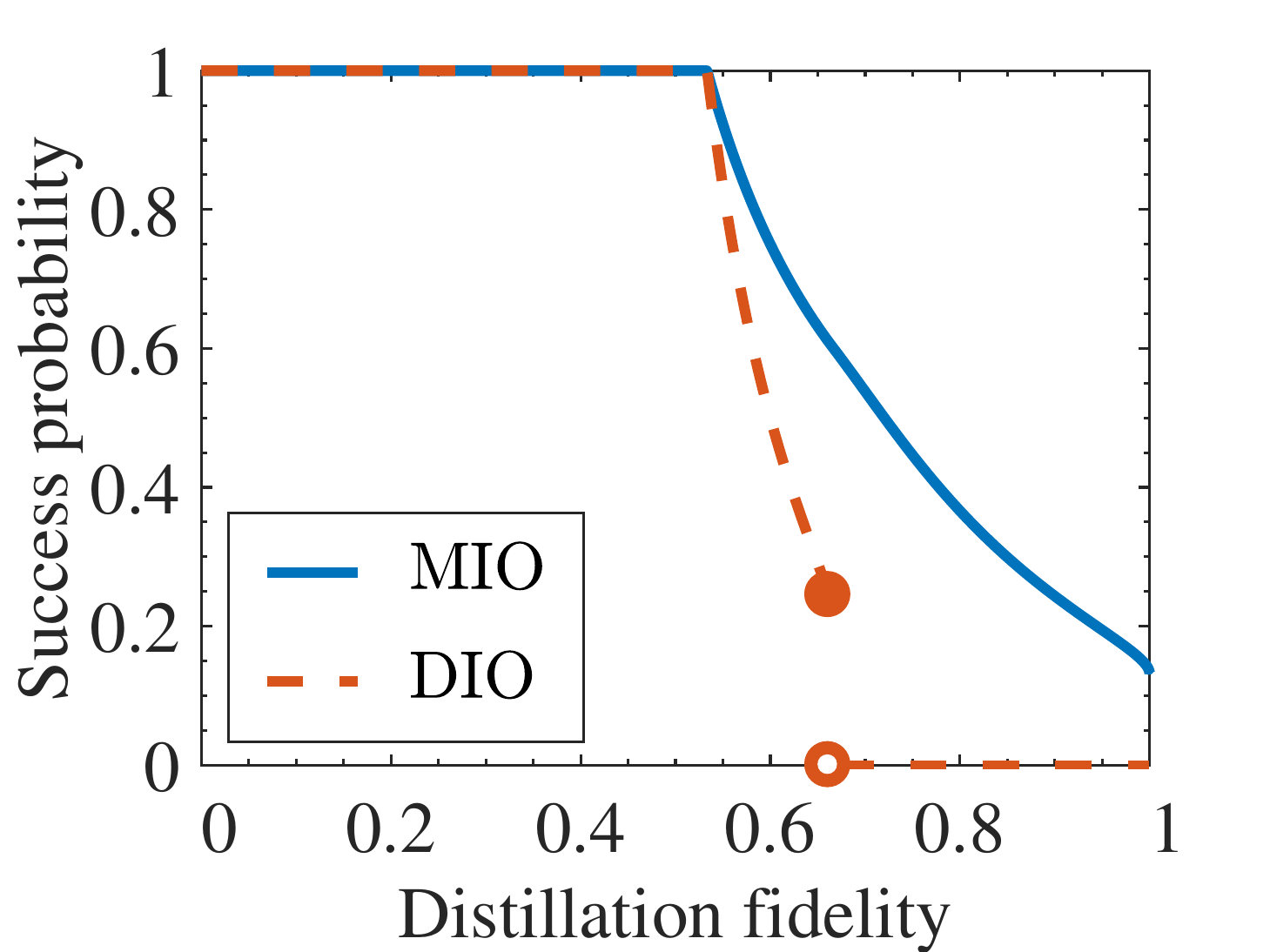}
\caption{Interplay between the fidelity $F=1-\ve$ and the success probability $p$ of coherence distillation for the example transformation $(\ket{0}+3\ket{1})/\sqrt{10} \!\to\! \Psi_3$. A discontinuity occurs at $F = 2/3$.}
\label{nontradeoff ex}
\end{figure}

\vspace{0.1cm}
\textbf{\textit{Probabilistic distillation with catalytic assistance.}}--- A more general coherence distillation setting is to consider the scenario with catalytic assistance~\cite{jonathan_1999}, where the input to the protocol consists of the resource state $\rho$ together with another state $\gamma$ (catalyst). As suggested by its name, we need to reproduce $\gamma$ untouched in the output regardless of whether the distillation process succeeds or not. In~\cite{bu_2016}, the authors studied catalytic coherence transformations without enforcing the preservation of the catalyst when the transformation fails --- it is then not surprising that catalytic assistance improves the success probability, since we take the risk to sacrifice our catalyst. However, we can show that using catalysts can enhance probabilistic distillation even when we require them to be reproduced regardless of the outcome.

Formally, we denote the catalysis-assisted maximal success probability of coherence distillation under the operation class $\cO$ as $P_{\cO}\big(\rho \xrightarrow{\gamma} \Psi_m, \ve\big)$, which is given by the maximal value of $p$ subject to the constraints
\begin{equation}\label{cata definition}
\begin{aligned}
\hspace{-0.2cm}\Pi(\rho\ox \gamma) =\  &\big(p \ket{0}\bra{0}  \ox \sigma + (1-p) \ket{1}\bra{1}\ox \o\big) \ox \gamma,\\
&F(\sigma,\Psi_m) \geq 1-\ve,\ \ \ \Pi
\in \cO.
\end{aligned}
\end{equation}
Since we can always choose not to interact with the catalyst, it is clear that $P_{\cO}\big(\rho \xrightarrow{\gamma} \Psi_m,\ve\big) \geq P_{\cO}(\rho\!\to\! \Psi_m, \ve)$.

Taking as an example the two-qubit state
$\rho=\frac{1}{2}(v_1+v_2)$ with
$\ket {v_1} = \frac{1}{2}(\ket {00} -\ket {01} -\ket {10}+\ket {11})$ and $\ket{v_2} = \frac{1}{5\sqrt{2}}(2\ket {00} +6\ket {01} -3\ket {10}+\ket {11})$,
it turns out that
the catalytic assistance of $\gamma = \Psi_2$ can enhance the success probability (by at least $12\%$) of distilling one coherent bit via DIO reliably ($\ve\le 0.01$) \cite{Note1}.
This example shows that the maximally coherent state can be used as a catalyst, manifesting a difference with the case of deterministic state transformation, where no transformation can be catalyzed by a maximally coherent state~\cite{jonathan_1999,du2015conditions}. We further note that if we allow a small perturbation of the catalyst to be returned in the protocol, one may obtain an even higher success probability as shown in the Supplemental Material \cite{Note1}  --- such a setting has been considered e.g.~in \cite{brandao_2015-2} for the resource theory of thermodynamics.

\vspace{0.1cm}
\textbf{\textit{Conclusions.}}---
We have developed a general framework of probabilistic coherence distillation. We interpreted the fundamental relations between the distillation fidelity and the maximal success probability via a gauge function construction, and showed that the maximal success probability under MIO and DIO can be efficiently computed via semidefinite programming.
We proved that distilling perfect coherence from any full-rank state is {\em impossible} even probabilistically, while any pure coherent state can always be perfectly distilled with MIO into a maximally coherent state of arbitrary dimension with a non-zero probability, highlighting an operational advantage of MIO over other sets of operations, in contrast with the deterministic case. On the other hand, we found that DIO provides no operational advantage over SIO in pure-state distillation, with the maximal achievable distillation probability being equal under the two classes of operations. We provided an analytical characterization of distillation with pure input states and in particular described the distillation of qubit maximally coherent states under MIO and DIO.
We further explored novel phenomena of coherence distillation such as the breakdown of the trade-off between  maximal success probability and fidelity under a certain threshold as well as the catalyst-assisted enhancement by maximally coherent states.

Our work establishes fundamental limitations to the processing of quantum coherence in realistic settings, opening new perspectives for its investigation and exploitation as a resource in quantum information processing and quantum technology \cite{hillery2016coherence,coles2016numerical,streltsov_2016,anshu2018quantum,diaz2018,lostaglio2015description,Frowis2011,streltsov_2017}. It would be of interest to analyze as well the task of probabilistic coherence dilution under different free operations, complementing the deterministic case studied in~\cite{zhao_2018}. Another interesting perspective for future work may be to apply the framework of probabilistic distillation developed here to the study of other important resource theories, such as those of asymmetry, magic, and thermodynamics \cite{chitambar_2018-1}.

\begin{acknowledgments}
\textbf{\textit{Acknowledgments.}}---
We would like to thank Gilad Gour, Filip Rozp\k{e}dek, Alexander Streltsov, Marco Tomamichel, and Yunlong Xiao for fruitful discussions. We acknowledge financial support from the European Research Council (ERC) under the Starting Grant GQCOP (Grant No.~637352).
\end{acknowledgments}

\bibliographystyle{apsrev4-1}
\bibliography{main}


\clearpage

\onecolumngrid
\begin{center}
\vspace*{.5\baselineskip}
{\textbf{\large Supplemental Material: \\[3pt] Probabilistic distillation of quantum coherence}}\\[1pt] \quad \\
\end{center}

\renewcommand{\theequation}{S\arabic{equation}}
\renewcommand{\thetheorem}{S\arabic{theorem}}
\setcounter{equation}{0}
\setcounter{figure}{0}
\setcounter{table}{0}
\setcounter{section}{0}

This supplemental material provides more detailed analysis and proofs of the results in the main text. We may reiterate some of the steps in the main text to make the supplemental material more explicit and self-contained.


\section{Proof of Proposition \ref{proposition 1}}

For any triplet $(\rho,m,\ve)$, the definition of the maximal success probability of coherence distillation is defined by 
\begin{subequations}\label{definition of pmax}
  \begin{align}
   P_{\cO}(\rho\!\to\! \Psi_m, \ve) = \max & \  p\\
   \text{s.t.} & \ \Pi_{A\to LB}(\rho) = p \ket{0}\bra{0}_L \ox \sigma + (1-p) \ket{1}\bra{1}_L\ox \o,\\
  &\ F(\sigma,\Psi_m) \geq 1-\ve,\ \Pi_{A\to LB} \in \cO,\ 0 \leq p \leq 1,
 \end{align}
\end{subequations}
where $F(\rho,\sigma):=\|\sqrt{\rho}\sqrt{\sigma}\|_1^2$ is the fidelity and $\|\cdot\|_1$ is the trace norm.
\begin{remark}
  The optimal solution in the optimization~\eqref{definition of pmax} can always achieve the equality in the condition $F(\rho,\Psi_m) \geq 1-\ve$. That is $P_{\cO}(\rho\!\to\! \Psi_m, \ve) = \widetilde P_{\cO}(\rho\!\to\! \Psi_m, \ve)$ with 
  \begin{subequations}\label{definition of pmax equality}
  \begin{align}
   \widetilde P_{\cO}(\rho\!\to\! \Psi_m, \ve) = \max & \  p\\
   \text{\rm s.t.} & \ \Pi_{A\to LB}(\rho) = p \ket{0}\bra{0}_L \ox \sigma + (1-p) \ket{1}\bra{1}_L\ox \o,\\
  &\ F(\sigma,\Psi_m) = 1-\ve,\ \Pi_{A\to LB} \in \cO,\ 0 \leq p \leq 1.
 \end{align}
\end{subequations}
It is clear that $P_{\cO}(\rho\!\to\! \Psi_m, \ve) \geq \widetilde P_{\cO}(\rho\!\to\! \Psi_m, \ve)$ since $P_{\cO}(\rho\!\to\! \Psi_m, \ve)$ is maximizing over a larger feasible set. However, for any optimal solution in \eqref{definition of pmax} such that $F(\sigma,\Psi_m) \geq 1-\ve$, we can further apply a depolarizing noise (free operation) to make the fidelity between $\rho$ and $\Psi_m$ decrease exactly to $1-\ve$, which gives a feasible solution for \eqref{definition of pmax equality}. Thus it also holds $P_{\cO}(\rho\!\to\! \Psi_m, \ve) \leq \widetilde P_{\cO}(\rho\!\to\! \Psi_m, \ve)$. In the following, we will equivalently use both optimizations \eqref{definition of pmax} and \eqref{definition of pmax equality}.
\end{remark}
\begin{remark}
  Denote the twirling operation $\cT(\rho)=\frac1{d!} \sum_{i=1}^{d!} P_i \rho P_i$, where $P_i$ are all the permutations on the input system and $d$ is the input dimension. If the distillation fails, we can perform a free operation $\cT\circ \Delta \in \SIO$  to make the unwanted state $\o$ completely mixed without changing the success probability. Thus, without loss of generality, we can take $\o=\1/m$. On the other hand, if the distillation process succeeds and we obtain an output state $\sigma$ such that $F(\sigma, \Psi_m) = 1-\ve$, we can further perform the free operation $\cT$ which gives a new output state $\cT(\sigma)$ always in the form of $a \Psi_m + b (\1-\Psi_m)/(m-1)$. Since $F(\sigma, \Psi_m) = 1-\ve$, we have $\cT(\sigma) = \Psi_m^\ve$ with $\Psi_m^\ve := (1-\ve)\Psi_m + \ve \frac{\1-\Psi_m}{m-1}$. Thus, without loss of generality, we can take $\sigma = \Psi_m^\ve$. Finally, the maximal success probability of coherence distillation can be given by
  \begin{subequations}\label{definition of pmax symmetry}
  \begin{align}
   P_{\cO}(\rho\!\to\! \Psi_m, \ve) = \max & \  p\\
   \text{\rm s.t.} & \ \Pi_{A\to LB}(\rho) = p \ket{0}\bra{0}_L \ox \Psi_m^\ve + (1-p) \ket{1}\bra{1}_L\ox \1/m,\\
  &\  \Pi_{A\to LB} \in \cO,\ 0 \leq p \leq 1,
 \end{align}
\end{subequations}
\end{remark}

\begingroup
\renewcommand{\theproposition}{\ref{proposition 1}}
\begin{proposition}
  For any triplet $(\rho,m,\ve)$ and operation class $\cO$, the maximal success probability $P_{\cO}(\rho\!\to\! \Psi_m,\ve)$ is given by
  \begin{align}\label{prop 1 equation}
\max \lsetr p\in \mathbb{R_+} \barr \cE(\rho) = p \cdot\Psi_m^\ve,\ \cE \in \subO \rsetr.
  \end{align}
As a consequence, it holds that $P_{\cO}\left(\rho\!\to\! \Psi_m,\ve\right)^{-1} = \min \lsetr t\in \mathbb{R_+} \barr \Psi_m^\ve \in t \cdot \cS_\rho \rsetr$ where $\cS_\rho \coloneqq \lsetr \cE(\rho) \barr \cE \in \subO \rsetr$ is the set of all the output operators of $\rho$ under the operation class $\subO$.
\end{proposition}
\begin{proof}
For any quantum operation $\Pi_{A\to LB}$ such that $\Pi_{A\to LB}(\rho) = \ket{0}\bra{0}_L \ox \cE_0(\rho) + \ket{1}\bra{1}_L \ox \cE_1(\rho)$ where $\cE_0$ and $\cE_1$ are two subnormalized operations, we can check that $\Pi_{A\to LB} \in \cO$ if and only if $\cE_0,\cE_1 \in \subO$ and $\cE_0 + \cE_1$ is trace preserving. Thus finding the optimal solution in the optimization \eqref{definition of pmax symmetry}
is equivalent to find the optimal subnormalized operations $\cE_0$ and $\cE_1$ such that $\cE_0(\rho) = p\cdot\Psi_m^\ve$, $\cE_1(\rho) = (1-p) \1/m$ and $\cE_0 + \cE_1$ trace-preserving. Since we can always take
$\cE_1(\rho) = (\tr \rho-\tr\cE_0(\rho)) \1/m$ without compromising the success probability, the maximal success probability of coherence distillation is only dependent on $\cE_0$, and the result follows.
\end{proof}
\endgroup

\section{Proof of Theorem \ref{SDP thm}}

\begingroup
\renewcommand{\theproposition}{\ref{SDP thm}}
\begin{theorem}
  For any triplet $(\rho,m,\ve)$, the maximal success probability of distillation under MIO and DIO are given by
  \begin{subequations}\label{MIO supp}
  \begin{align}
 \hspace*{-23pt}P_{\MIO}(\rho\!\to\! \Psi_m,\ve) = \underset{G,C}{\text{\rm maximize}} &\quad \tr G\rho\\
  \text{\rm subject to} &\quad \ \Delta(G) = m \Delta(C), \\
  &\quad \ 0 \leq C \leq G \leq \1,\\
  &\quad \tr C\rho \geq (1-\ve) \tr G \rho,
  \end{align}
  \end{subequations}
  and
\begin{subequations}\label{DIO supp}
  \begin{align}
 \hspace*{-23pt}P_{\DIO}(\rho\!\to\! \Psi_m,\ve) = \underset{G,C}{\text{\rm maximize}} &\quad \tr G\rho\\
  \text{\rm subject to} &\quad \ \Delta(G) = m \Delta(C), \\
  &\quad \ 0 \leq C \leq G \leq \1,\\
  &\quad \tr C\rho \geq (1-\ve) \tr G \rho,\\
  & \quad \ G = \Delta(G).
  \end{align}
  \end{subequations}
\end{theorem}
\begin{proof}
We show the proof for MIO first. Denote $J_{\cN}$ as the \Choi matrix of operation $\cN$. Due to the \Choi isomorphism, we can write the optimization \eqref{prop 1 equation} in terms of \Choi matrix,
\begin{subequations}\label{MIO Choi matrix form}
  \begin{align}
    P_{\MIO}(\rho \to \Psi_m,\ve) = \max &\ p\\
    \text{s.t.} & \tr_A J_{\cE} (\rho^T\ox \1_B) = p \cdot \Psi_m^\ve,\\
    &\ J_{\cE} \geq 0, \tr_B J_{\cE} \leq \1_A,\\
    & \tr_A J_{\cE} (\ket{i}\bra{i}^T \ox \1_B) = \Delta(\tr_A J_{\cE} (\ket{i}\bra{i}^T \ox \1_B)),\, \forall i.
  \end{align}
\end{subequations}
For any optimal subnormalized quantum operation $\cE$ in optimization~\eqref{prop 1 equation}, the operation $\widetilde \cE = \cT\circ \cE$ is also optimal since $\Psi_m^\ve$ is invariant under the twirling operation $\cT$. Thus $J_{\widetilde \cE}$ admits the structure of $J_{\widetilde \cE} = C^T \ox \Psi_m + D^T \ox (\1 - \Psi_m)$ for some operators $C$ and $D$. Taking this form of \Choi matrix into the conditions of the optimization~\eqref{MIO Choi matrix form}, we have 
\begin{subequations}
  \begin{align}
    P_{\MIO}(\rho \to \Psi_m,\ve) = \max &\ p\\
    \text{s.t.} & \tr C\rho = p(1-\ve),\, \tr (m-1)D \rho = p\ve,\\
    &\ C \geq 0,\, D\geq 0,\, C+(m-1)D \leq \1,\\
    &\ \Delta(C) = \Delta(D).
  \end{align}
\end{subequations}
Denoting $G := C+(m-1)D$ and eliminating variable $D$, we will obtain the result \eqref{MIO supp}. Note that the last condition in \eqref{MIO supp} can be taken as equality due to the argument in Remark 1. Following the similar steps, we can obtain the optimization for DIO.
\end{proof}
\endgroup

\begin{remark}
  From the conditions $G = \Delta(G) = m \Delta(C)$, we can eliminate the variable $G$ in the optimization \eqref{DIO supp} and obtain
    \begin{subequations}\label{DIO no G}
  \begin{align}
 P_{\DIO}(\rho\!\to\! \Psi_m,\ve) = \underset{C}{\text{\rm maximize}} &\quad \tr m \Delta(C) \rho\\
  \text{\rm subject to} &\quad \ 0 \leq C \leq m \Delta(C) \leq \1, \\ &\quad \ \tr C\rho \geq m (1-\ve) \tr \Delta(C) \rho.
  \end{align}
  \end{subequations}
\end{remark}

\begin{remark}
We will make repeated use of the semidefinite programs given in Theorem~\ref{SDP thm} and their Lagrange duals. We report the dual forms below. The dual of \eqref{MIO supp} is given by
\begin{equation}\label{eq:d_MIO}
\begin{aligned}
 P_\MIO(\rho\!\to\! \Psi_m, \ve) =\;\, & \underset{X,Y,Z,\lambda}{\text{\rm minimize}} & & \Tr X\\
& \text{\rm subject to} & & \left[ 1 - \lambda (1-\ve) \right] \rho + Y - X + \Delta(Z) \leq 0,\\
&&& \lambda \rho - Y - m \Delta(Z) \leq 0,\\
&&& X,Y \geq 0,\\
&&& \lambda \in \RR_+,\\
\end{aligned}
\end{equation}
The dual of \eqref{DIO no G} is given by
\begin{equation}\label{eq:d_DIO}
\begin{aligned}
 P_\DIO(\rho\!\to\! \Psi_m, \ve) =\;\, & \underset{X,Y,\lambda}{\text{\rm minimize}} & & \Tr X\\
& \text{\rm subject to} & & m \Delta(\rho) + m\Delta(Y) - Y - m\Delta(X) + \lambda \rho - m (1-\ve) \lambda \Delta(\rho) \leq 0,\\
&&& X,Y \geq 0,\\
&&& \lambda \in \RR_+.
\end{aligned}
\end{equation}
To see that strong duality holds, it suffices to note the existence of strictly feasible solutions: in~\eqref{eq:d_MIO}, take $\lambda = \frac{1}{1-\ve}, X = 3\lambda \id, Y = 2\lambda \id, Z = 0$, and in~\eqref{eq:d_DIO}, take $\lambda = \frac{1}{1-\ve}, X = 3\lambda \id, Y = \id$.
\end{remark}

\begin{remark}
In the case $\ve=0$, the primal problems can alternatively be expressed as
\begin{equation}
\begin{aligned}
 P_{\MIO}(\rho\!\to\! \Psi_m, 0) =\;\, & \underset{G,C}{\text{\rm maximize}} & & \Tr G \rho\\
& \text{\rm subject to} & & C \geq 0,\\
&&& G \leq \1 ,\\
&&& \Delta(G) = m \Delta(C),\\
&&& G - C \in \big\{ A \ \big| \ \supp(A) \subseteq \ker(\rho) \big\} \cap \HH_+.
\end{aligned}
\end{equation}
\begin{equation}
\begin{aligned}
P_\DIO(\rho\!\to\! \Psi_m, 0) =\;\, & \underset{C}{\text{\rm maximize}} & & \Tr m \Delta(C) \rho\\
& \text{\rm subject to} & & C \geq 0,\\
&&& m\Delta(C) \leq \1,\\
&&& m\Delta(C) - C \in \big\{ A \ \big| \ \supp(A) \subseteq \ker(\rho) \big\} \cap \HH_+,
\end{aligned}
\end{equation}
where we have used $\HH_+$ to denote the set of positive semidefinite matrices. This gives the duals
\begin{equation}
\begin{aligned}
 P_\MIO(\rho\!\to\! \Psi_m, 0) =\;\, & \underset{W,Z}{\text{\rm minimize}} & & \Tr ( W + Z ) + 1\\
& \text{\rm subject to} & & \rho + W + \Delta(Z) \geq 0,\\
&&& W + m\Delta(Z) \geq 0,\\
&&& W \in \big\{ A \ \big| \ \supp(A) \subseteq \supp(\rho) \big\} + \HH_+,
\end{aligned}
\end{equation}
\begin{equation}
\begin{aligned}
 P_{\DIO} (\rho\!\to\! \Psi_m, 0) =\;\, & \underset{X,W}{\text{\rm minimize}} & & \Tr X\\
& \text{\rm subject to} & & m \Delta(\rho) + W - m\Delta(W) - m \Delta(X) \leq 0,\\
&&& X \geq 0,\\
&&& W \in \big\{ A \ \big| \ \supp(A) \subseteq \supp(\rho) \big\} + \HH_+.
\end{aligned}
\end{equation}
\end{remark}

\section{Proof of Theorem \ref{full rank thm}}

To prove Theorem~\ref{full rank thm} in the main text, we will establish a stronger version of the result as follows.

\begingroup
\renewcommand{\theproposition}{S1}
\begin{theorem}
For any triplet $(\rho,m,0)$ with full-rank state $\rho$, it holds that $P_{\MIO}(\rho\!\to\! \Psi_m,0) = 0$. For any triplet $(\varphi,m,0)$ with coherent pure state $\ket{\varphi} = \sum_{i=1}^n \varphi_i\ket{i}$ and $\varphi_i \neq 0$, $n\geq 2$, it holds that
\begin{align}\label{MIO lower bound}
P_{\MIO}(\varphi \to \Psi_m,0) \geq \frac{n^2}{\sum_{i=1}^n |\varphi_i|^{-2}} \left\| \frac{n-m}{n-1} \widetilde{\varphi} + \frac{n(m-1)}{n-1} \Delta(\widetilde{\varphi}) \right\|_{\infty}^{-1} \geq \frac{n^2}{m(\sum_{i=1}^n |\varphi_i|^{-2})} > 0
\end{align}
where \begin{align}
\ket{\widetilde \varphi} := \frac{1}{\sqrt{s}}\sum_{i=1}^n \frac{\varphi_i}{|\varphi_i|^2}\ket{i} \quad \text{\rm with} \quad s =  \sum_{j=1}^n |\varphi_j|^{-2}.
\end{align}
\end{theorem}
\begin{proof}
Recall the SDP
\begin{align}\label{sSDP MIO}
 P_{\MIO}(\rho\!\to\! \Psi_m,\ve) = \max \left\{\tr G\rho \ \big|\  \Delta(G) = m \Delta(C), \ 0 \leq C \leq G \leq \1,
 \tr C\rho = (1-\ve) \tr G \rho\right\}.
\end{align}
For the first argument, we know that $G-C \geq 0$ and $\tr (G - C)\rho = 0$. Since $\rho$ is full-rank, we have $G = C$. Together with $\Delta(G) = m \Delta(C)$, we have $G = C = 0$ and $P_{\MIO}(\rho\!\to\! \Psi_m,0) = 0$.

As for the second argument, let us choose
\begin{align}
C = c \widetilde{\varphi}\, ,\qquad G = c \widetilde{\varphi} + \frac{(m-1)c}{n-1} \left(n\Delta(\widetilde{\varphi}) - \widetilde{\varphi} \right),
\end{align}
where
\begin{align}
c = \left\| \frac{n-m}{n-1} \widetilde{\varphi} + \frac{n(m-1)}{n-1} \Delta(\widetilde{\varphi}) \right\|_{\infty}^{-1} .
\end{align}
We check the constraints in~\eqref{sSDP MIO} one by one. The first condition trivially holds by the construction. The last condition holds since $\<\varphi|n\Delta(\widetilde \varphi) - \widetilde \varphi|\varphi\> = 0$, which implies that $\<\varphi|C|\varphi\> = \<\varphi|G|\varphi\>$. We now move on to the second condition. Clearly $C\geq 0$  and furthermore $G\geq C$ as follows from $\varphi \leq n\Delta(\varphi)$. To show that $G\leq \1$, just observe that
\begin{align}
\|G\|_{\infty} = c \left\| \widetilde{\varphi} + \frac{m-1}{n-1} \left(n\Delta(\widetilde{\varphi}) - \widetilde{\varphi} \right) \right\|_{\infty} = c \left\| \frac{n-m}{n-1} \widetilde{\varphi} + \frac{n(m-1)}{n-1}\Delta(\widetilde{\varphi})\right\|_{\infty} = 1\, .
\end{align}
Hence, $C,G$ as defined above form a valid ansatz for the semidefinite program~\eqref{sSDP MIO} and
\begin{align}
P_{\text{MIO}}(\varphi, m,0) \geq \tr G \varphi = \frac{n^2 c}{s},
\end{align}
which yields the first lower bound in~\eqref{MIO lower bound}. As for the second bound, it suffices to show that $c\geq 1/m$, i.e. that $c^{-1}\leq m$. This can be done thanks to the triangle inequality:
\begin{align}
c^{-1} = \left\| \widetilde{\varphi} + \frac{m-1}{n-1} \left( n\Delta(\widetilde{\varphi}) - \widetilde{\varphi} \right) \right\|_{\infty} \leq \left\| \widetilde{\varphi} \right\|_{\infty} + (m-1) \left\| \frac{n\Delta(\widetilde{\varphi}) - \widetilde{\varphi}}{n-1} \right\|_{\infty} \leq 1 + (m-1) = m\, ,
\end{align}
where we have used the fact that $\frac{n \Delta(\widetilde{\varphi}) - \widetilde{\varphi}}{n-1}$ is a valid density matrix.
\end{proof}
\endgroup


\section{Proof of Theorem \ref{thm:dio sio}}

\begingroup
\renewcommand{\theproposition}{\ref{thm:dio sio}}
\begin{theorem}
For any pure state $\varphi$ and any $m$, it holds that
\begin{equation}\begin{aligned}
  P_\DIO(\varphi \!\to\! \Psi_m,0) = P_\SIO(\varphi \!\to\! \Psi_m,0).
\end{aligned}\end{equation}
\end{theorem}
\begin{proof}
We will assume without loss of generality that the coefficients of $\ket\varphi$ are non-negative and arranged in non-increasing order. Let
\begin{equation}\begin{aligned}
  k \coloneqq \argmin_{1 \leq j \leq m} \, \frac{1}{j} \sum_{i=m-j+1}^d \varphi_i^2
\end{aligned}\end{equation}
such that $P_\SIO(\varphi \!\to\! \Psi_m, 0) = \frac{m}{k} \sum_{i=m-k+1}^d \varphi_i^2$. Note that if $k = m$, then $1 \geq P_\DIO(\varphi \!\to\! \Psi_m , 0) \geq P_\SIO(\varphi \!\to\! \Psi_m, 0) = 1$. In the following, we will therefore assume that $1 \leq k \leq m-1$. Further, let us begin by considering strictly positive $\ve$.

Define
\begin{equation}\begin{aligned}
  \varphi_A &\coloneqq \sum_{i,j=1}^{m-k} \proj{i} \varphi \proj{j}\\
  \varphi_B &\coloneqq \sum_{i,j=m-k+1}^{d} \proj{i} \varphi \proj{j}
\end{aligned}\end{equation}
and notice that $P_\SIO(\varphi \!\to\! \Psi_m, 0) = \frac{m}{k} \Tr \varphi_B$.

Recall that the SDP for the maximum distillation probability under DIO is given by
\begin{subequations}\label{eq:dual_DIO}
\begin{align}
 P_\DIO(\rho\!\to\! \Psi_m, \ve) = \underset{X,P,\lambda}{\text{minimize}} &\quad \Tr X\nonumber\\
\text{subject to} &\quad m \Delta(\rho) + m\Delta(P) - P - m\Delta(X) + \lambda \rho - m (1-\ve) \lambda \Delta(\rho) \leq 0\label{SDP_cond}\\\
&\quad P,X \geq 0\\
&\quad \lambda \in \RR_+.
\end{align}
\end{subequations}
We will take the ansatz
\begin{equation}\begin{aligned}\label{eq:ansatz}
  P &= \mu_1  \varphi_B - \mu_2  \varphi_A + \left[ (1-\ve)\lambda - 1 \right] \varphi\\
  X &= \mu_1 \Delta(\varphi_B)\\
\end{aligned}\end{equation}
for some coefficients $\mu_1, \mu_2 \in \RR_+$. Our aim will now be to show that there exists a suitable choice of $\lambda$, $\mu_1$, and $\mu_2$ such that $(X,P,\lambda)$ is feasible for the SDP~\eqref{eq:dual_DIO}.

To this end, denote by $\Theta$ the all-ones matrix of appropriate size, and notice that the condition~\eqref{SDP_cond} reduces to
\begin{equation}\begin{aligned}
  0 &\geq m \Delta(\varphi) + m\Delta(P) - P - m\Delta(X) + \lambda \varphi - m (1-\ve) \lambda \Delta(\varphi)\\
  &= \mu_2 \left( \varphi_A - m \Delta(\varphi_A) \right) - \mu_1 \varphi_B + (1+\lambda\ve) \varphi\\
  &= \begin{pmatrix} \left[ \mu_2 + 1 + \lambda\ve \right] \Theta_A - m\mu_2 \id_A & \left[1+\lambda\ve\right]  \Theta_O\\
                      \left[1+\lambda\ve\right] \Theta_O^T & \left[ 1 + \lambda\ve - \mu_1 \right] \Theta_B \end{pmatrix} \circ \varphi\\
  &\eqqcolon Q \circ \varphi
\end{aligned}\end{equation}
where $\circ$ denotes the Schur product. Showing that $Q \leq 0$ will then imply the desired relation $Q \circ \varphi \leq 0$ by the Schur product theorem. From the generalized Schur complement condition, $Q \leq 0$ if and only if the following all hold~\cite{Zhang2005,Bhatia2007}:
\begin{equation}\begin{aligned}
 \cond{i} &\quad \left( \mu_2 + 1 + \lambda\ve \right) \Theta_A - m\mu_2 \id_A \leq 0,\\
 \cond{ii} &\quad \left( 1 + \lambda\ve - \mu_1 \right)\Theta_B \leq 0,\\
 \cond{iii} &\quad \left(\mu_2 + 1 + \lambda\ve \right) \Theta_A - m\mu_2 \id_A - (1+\lambda\ve)^2 \Theta_O \left( \left[ 1 + \lambda\ve - \mu_1 \right] \Theta_B \right)^{-1} \Theta_O^T \leq 0
\end{aligned}\end{equation}
where $M^{-1}$ denotes the Moore-Penrose pseudoinverse, which in particular satisfies $M^{-1} = M/\Tr(M)^2$ for any rank-one Hermitian $M$~\cite{ben-israel_2003}. From the first two conditions, we have
\begin{equation}\begin{aligned}
  \cond{i} &\quad \mu_2 \geq (1+\lambda\ve)\frac{m-k}{k},\\
  \cond{ii} &\quad \mu_1 \geq 1 + \lambda\ve,
\end{aligned}\end{equation}
and the third condition reduces to
\begin{equation}\begin{aligned}
 0 &\geq \left(\mu_2 + 1 + \lambda\ve \right) \Theta_A - m\mu_2 \id_A - \frac{(1+\lambda\ve)^2}{\left[ 1 + \lambda\ve - \mu_1 \right] \Tr\Theta_B^2} \Theta_O \Theta_B \Theta_O^T\\
  &= \left(\mu_2 + 1 + \lambda\ve \right) \Theta_A - m\mu_2 \id_A - \frac{(1+\lambda\ve)^2}{ 1 + \lambda\ve - \mu_1} \Theta_A\\
  &= \left[ \mu_2 - \frac{ (1+\lambda\ve) \mu_1 }{1 + \lambda\ve - \mu_1} \right] \Theta_A - m\mu_2 \id_A.
\end{aligned}\end{equation}
Noting that the coefficient $\mu_2 - \frac{ (1+\lambda\ve) \mu_1 }{1 + \lambda \ve - \mu_1}$ can never be negative when conditions (i) and (ii) are satisfied, we have
\begin{equation}\begin{aligned}
  \cond{iii} & \quad m \mu_2  \geq (m-k) \left[ \mu_2 - \frac{ (1+\lambda\ve) \mu_1 }{1 + \lambda\ve - \mu_1} \right].
\end{aligned}\end{equation}

To ensure that $P\geq 0$, notice that since
\begin{equation}\begin{aligned}
  P = \begin{pmatrix} \left[-\mu_2 + (1-\ve) \lambda - 1 \right] \Theta_A & \left[(1-\ve)\lambda - 1\right] \Theta_O \\
  \left[(1-\ve)\lambda - 1\right] \Theta_O^T & \left[\mu_1 + (1-\ve) \lambda - 1 \right] \Theta_B \end{pmatrix} \circ \varphi,
\end{aligned}\end{equation}
the positivity of $P$ is equivalent to the positivity of the matrix
\begin{equation}\begin{aligned}
  P' \coloneqq \begin{pmatrix} -\mu_2 + (1-\ve) \lambda - 1 & (1-\ve)\lambda - 1 \\ (1-\ve)\lambda - 1 & \mu_1 + (1-\ve) \lambda - 1 \end{pmatrix}.
\end{aligned}\end{equation}
We therefore have two additional conditions:
\begin{equation}\begin{aligned}
  \cond{iv} &\quad 0 \leq \Tr P' = 2 (1- \ve) \lambda + \mu_1 - \mu_2 - 2\\
  \cond{v} &\quad 0 \leq \Det P' = \mu_2 - (1-\ve) \lambda \mu_2 - \mu_1 \left[ 1 - (1-\ve) \lambda + \mu_2 \right].
\end{aligned}\end{equation}
Let us stress that conditions (i)--(v) together with $\mu_1 \geq 0$ constitute a set of sufficient conditions for a given choice of $(X,P,\lambda)$ of the form~\eqref{eq:ansatz} to satisfy the feasibility conditions of SDP~\eqref{eq:dual_DIO}.

We will now make the ansatz
\begin{equation}\begin{aligned}
\lambda &= \frac{\mu_1 - \mu_2}{\ve \mu_2}\\
\mu_1 &= \frac{k \mu_2 (\mu_2 - 1)}{m - k}
\end{aligned}\end{equation}
leaving $\mu_2$ as a free variable for now. With this choice, conditions (i)--(iii) are always satisfied for any $\mu_2 \in \RR$, while the other conditions reduce to
\begin{equation}\begin{aligned}
   \cond{iv} &\quad 2(k\mu_2 - m) + \ve (2 k - (2k+m) \mu_2 + k\mu_2 ^2) \geq 0\\
   \cond{v} &\quad \mu_2 \left( \ve k m (1 - \mu_2)^2 - (m - k \mu_2)^2 \right) \leq 0
\end{aligned}\end{equation}
We will now make a choice of $\mu_2$ which can be verified to satisfy the above the inequalities for any $\ve < \frac{k}{m}$ as
\begin{equation}\begin{aligned}
\mu_2 &= \frac{m + m^2\sqrt{\ve}}{k - \ve m}
\end{aligned}\end{equation}
where we note that $\mu_2 \geq 1$ and $\lim_{\ve \to 0} \mu_1 = \lim_{\ve \to 0} \mu_2 = \frac{m}{k}$. Since all conditions (i)--(v) are satisfied for our choice of $(X,P,\lambda)$ with the given $\mu_1$ and $\mu_2$ for any $\ve < \frac{k}{m}$, the triple $(X,P,\lambda)$ is a valid feasible solution for the SDP~\eqref{eq:dual_DIO}. This means in particular that for any $0 < \ve < \frac{k}{m}$ we have
\begin{equation}\begin{aligned}
  P_\DIO(\varphi \!\to\! \Psi_m,\ve) \leq \Tr X = \mu_1 \Tr \varphi_B.
\end{aligned}\end{equation}
Using the fact that $\SIO \subset \DIO$ as well as that $P_\DIO(\varphi \!\to\! \Psi_m,\ve) \geq P_\DIO(\varphi \!\to\! \Psi_m,0)$ for any $\ve$, this then gives
\begin{equation}\begin{aligned}
  P_\SIO(\varphi \!\to\! \Psi_m,0) &\leq P_\DIO(\varphi \!\to\! \Psi_m,0)\\
                   &\leq \lim_{\ve \to  0} P_\DIO(\varphi \!\to\! \Psi_m,\ve)\\
                   &\leq \frac{m}{k} \Tr\varphi_B\\
                   &= P_\SIO(\varphi \!\to\! \Psi_m,0)
\end{aligned}\end{equation}
which completes the proof.
\end{proof}
\endgroup


\section{Proof of Proposition \ref{pure to psi 2} and \ref{nontradeoff prop}}

\begingroup
\renewcommand\theproposition{\ref{pure to psi 2}}
\begin{proposition}
For $\cO\in \{\text{DIO}, \text{MIO}\}$ and any pure state $\varphi$ with $\varphi_1 \geq ... \geq \varphi_n > 0$, it holds that
\begin{equation}
  P_{\cO}(\varphi\!\to\!\Psi_2,\ve) \!=\! \left\{ \!\begin{array}{ll} 1 & \text{if $\ve\!\geq\! \ve_0(\varphi_1)$,} \\  2(1-\varphi_1^2) \left(\frac{\sqrt{1-\ve}+\sqrt{\ve}}{1-2\ve}\right)^2 & \text{otherwise,} \end{array} \right.
\end{equation}
where
\begin{align}\begin{displaystyle}\ve_0(\varphi_1) \coloneqq \left\{ \begin{array}{ll} 0 &\quad \text{if $\varphi_1\leq\frac{1}{\sqrt2}$,} \\ \frac12-\varphi_1\sqrt{1-\varphi_1^2} & \quad \text{otherwise.} \end{array}\right.
\end{displaystyle}
\end{align}
\end{proposition}
\begin{proof}
Since the case $\ve = 0$ follows from Theorem~\ref{thm:dio sio}, we will assume $\ve > 0$.  For the case of $\varphi_1 \leq \frac{1}{\sqrt{2}}$, we have
\begin{align}
  P_{\cO}(\varphi\!\to\!\Psi_2,\ve) \geq P_{\SIO}(\varphi\!\to\!\Psi_2,0) = 1.
\end{align}
 From the result in~\cite{Regula_2017-1}, we know that if $\ve \geq  \frac12 - \varphi_1\sqrt{1-\varphi_1^2}$, then $P_{\DIO}(\varphi\!\to\!\Psi_2,\ve) = 1 $. In the following, we therefore only consider the case $\ve <  \frac12 - \varphi_1\sqrt{1-\varphi_1^2}$.  We prove this result by explicit constructing feasible solutions in both primal and dual SDPs.
  The primal SDP under DIO is given by
  \begin{align}
  P_{\DIO}(\varphi \!\to\! \Psi_2,\ve) = \max \left\{\tr G\varphi\ \big|\ \tr C\varphi \geq (1-\ve) \tr G \varphi,\ 0 \leq C \leq G \leq \1,\ \Delta(G) = 2 \Delta(C),\ G = \Delta(G)\right\}.
  \end{align}
We take the ansatz
\begin{align}
  G = \1 - x\ket{0}\bra{0},\quad C = \frac12 G + y \sum_{i=2}^n \varphi_1\varphi_n(\ket{1}\bra{i}+\ket{i}\bra{1}).
\end{align}
Then we have
\begin{align}
  \tr G\rho = 1-x\varphi_1^2, \quad \tr C \rho  =  \frac{1}{2}(1-x\varphi_1^2)+2y\varphi_1^2(1-\varphi_1^2),
\end{align}
and
\begin{align}
  \text{spec}(C) = \text{spec}(G-C)= \left\{
  \underbrace{\frac12, \cdots, \frac12}_{n-2 \ \text{fold}},
  \frac12 - \frac{x}{4} - \sqrt{y^2\varphi_1^2(1-\varphi_1^2)+\frac{x^2}{16}}, \frac12 - \frac{x}{4} + \sqrt{y^2\varphi_1^2(1-\varphi_1^2)+\frac{x^2}{16}}\right\}.
\end{align}
Then we have the relaxation
\begin{equation}
  \begin{aligned}
    P_{\DIO}(\varphi\!\to\!\Psi_2,\ve) \geq \text{\rm maximize} &\quad 1-x \varphi_1^2\\
    \text{subject to} &\quad  4y \varphi_1^2(1-\varphi_1^2) = (1-2\ve) (1-x\varphi_1^2),\\
    &\quad  1-x \geq 4y^2\varphi_1^2(1-\varphi_1^2),\\
    &\quad 0 \leq x \leq 1.
  \end{aligned}
\end{equation}
By choosing
\begin{align}
  x = \frac{1-2(1-\varphi_1^2)(\frac{\sqrt{1-\ve}+\sqrt{\ve}}{1-2\ve})^2}{\varphi_1^2},\quad y = \frac{(\sqrt{1-\ve}+\sqrt{\ve})^2}{2\varphi_1^2(1-2\ve)},
\end{align}
we can verify that this is a feasible solution. Thus we have
\begin{align}\label{tmp1}
P_{\DIO}(\varphi\!\to\!\Psi_2,\ve) \geq 2(1-\varphi_1^2)\left(\frac{\sqrt{1-\ve}+\sqrt{\ve}}{1-2\ve}\right)^2.
\end{align}

As for the dual problem, we consider the dual SDP under MIO,
  \begin{equation}
  \begin{aligned}
  P_{\MIO}(\varphi\!\to\! \Psi_2,\ve) = \text{\rm minimize} &\quad \tr Y\\
  \text{\rm subject to} &\quad (1-(1-\ve)x)\varphi + X  + \Delta(Z) \leq Y,\\
  &\quad x\varphi - X - 2\Delta(Z) \leq 0,\\
  &\quad x \geq 0, X \geq  0, Y \geq 0
  \end{aligned}
  \end{equation}
Taking
\begin{gather}
    x = \frac{(\sqrt{1-\ve}+\sqrt{\ve})^2}{\sqrt{1-\ve}\sqrt{\ve}(1-2\ve)}, \quad Y = 2\left(\frac{\sqrt{1-\ve}+\sqrt{\ve}}{1-2\ve}\right)^2(\1-\ket{0}\bra{0})\varphi(\1-\ket{0}\bra{0}),\\
     Z = \frac{2}{1-2\ve}\ket{0}\bra{0}\varphi\ket{0}\bra{0},\quad
    X = Y - \Delta(Z) - (1-(1-\ve)x)\varphi,
  \end{gather}
  we can verify that $\{x,X,Y,Z\}$ is a valid feasible solution. Thus
  \begin{align}\label{tmp2}
    P_{\MIO}(\varphi\!\to\! \Psi_2,\ve) \leq 2(1-\varphi_1^2)\left(\frac{\sqrt{1-\ve}+\sqrt{\ve}}{1-2\ve}\right)^2.
  \end{align}
  Combining Eqs.~\eqref{tmp1} and~\eqref{tmp2}, we have the desired result.
\end{proof}
\endgroup


\begingroup

Recall that IO and SIO have the same power in pure-state transformations and it holds that
\begin{align}\label{eq:vidal supp}
P_\SSIO(\varphi\!\to\! \Psi_m,0) \!=\! \begin{cases} 0  &  \text{ if } \rank\Delta(\varphi) < m,\\
\begin{displaystyle}\min_{k\in[1,m]} \frac{m}{k} \hspace{-.3em}\sum_{i=m\!-\!k\!+\!1}^d \hspace{-.3em}\varphi_i^2 \end{displaystyle} \;\; &\text{ otherwise.}\end{cases}\raisetag{2.2\baselineskip}
\end{align}

\renewcommand\theproposition{\ref{nontradeoff prop}}
\begin{proposition}
  For any pure state $\ket{\varphi} = \sum_{i=1}^n \varphi_i \ket{i}$ with  $\varphi_i >0 $, it holds that
\begin{itemize}
  \item If $n \geq m$, $P_{\DIO}(\varphi\!\to\! \Psi_m,\ve) > 0$.
  \item If $n < m$, $P_{\DIO}(\varphi\!\to\! \Psi_m,\ve) \begin{cases} >0, \quad \ve \geq 1 - \frac{n}{m}, \\  = 0, \quad \ve < 1 - \frac{n}{m}.\end{cases}$
\end{itemize}
\end{proposition}
\begin{proof}
  For the first argument, if $n \geq m$, we know that
  \begin{align}
  P_{\DIO}(\varphi\!\to\! \Psi_m,\ve) \geq P_{\SIO}(\varphi\!\to\! \Psi_m,0) > 0,
  \end{align}
  where the second inequality follows from Eq.~\eqref{eq:vidal supp}.

Note that if $P_{\cO}(\sigma_1 \!\to\! \sigma_2,0) = 1$, then $P_{\cO}(\rho \!\to\! \sigma_2,0) \geq P_{\cO}(\rho \!\to\! \sigma_1,0)$ since we can first transform $\rho$ to $\sigma_1$ perfectly and then to get $\sigma_2$.
For the second argument, if $\ve \geq 1- \frac{n}{m}$, we have
  \begin{align}\label{tmp 1}
    P_{\DIO}(\Psi_n\!\to\!\Psi_m^\ve,0) = P_{\DIO}(\Psi_n\!\to\! \Psi_m,\ve) = 1.
  \end{align}
The first equality follows from the fact that $P_{\cO}(\rho\!\to\!\Psi_m,\ve) = P_{\cO}(\rho\!\to\!\Psi_m^\ve,0)$. The second equality follows from Lemma~\ref{psin to psim} below.
Then
\begin{equation}
  \begin{aligned}
  P_{\DIO}(\varphi\!\to\! \Psi_m,\ve) &= P_{\DIO}(\varphi\!\to\! \Psi_m^{\ve},0)\\ &\geq P_{\DIO}(\varphi\!\to\! \Psi_n,0) \\&\geq P_{\SIO}(\varphi\!\to\! \Psi_n,0)\\&> 0.
  \end{aligned}
  \end{equation}
   The first inequality follows from Eq.~\eqref{tmp 1}. The last inequality follows from Eq.~\eqref{eq:vidal supp}. If $ \ve \leq 1-\frac{n}{m}$, we have $P_{\DIO}(\varphi\!\to\! \Psi_m,\ve) \leq P_{\DIO}(\Psi_n\!\to\!\Psi_m,\ve) = 0$, where the second equality follows from Lemma~\ref{psin to psim}.
\end{proof}

\endgroup


\begingroup
\renewcommand\theproposition{S2}
\begin{lemma}\label{psin to psim}
    For any integer $n \leq m$, it holds that
    \begin{align}
    P_{\DIO}(\Psi_n\!\to\! \Psi_m,\ve) = \begin{cases} 1, \quad \ve \geq 1 - \frac{n}{m}, \\ 0, \quad \ve < 1 - \frac{n}{m}.\end{cases}
    \end{align}
\end{lemma}
\begin{proof}
  For $\ve \geq 1 - \frac{n}{m}$, we can take feasible solution $G = \1$, $C = \frac{n}{m}\psi_n$, which gives feasible value $1$ in the primal problem. For $\ve < 1 - \frac{n}{m}$, we can take feasible solution $x = \frac{1}{1-\frac{n}{m}-\ve}$, $X = Y = 0$, $Z = \frac{1}{m-n-m\ve}\1$, $W =  \frac{1}{m-n-m\ve}(n\psi_n - \1)$, which gives feasible value $0$ in the dual problem.
\end{proof}
\endgroup

\section{Examples for catalyst-assisted enhancement}

For the catalysis scenario, we require that the catalyst is returned no matter the distillation process succeeds or not. In a more general case than the setting presented in the main text, we may accept imperfect catalyst to be returned. Denote $\delta$ as the infidelity tolerance for the returning catalyst. We define the catalysis-assisted maximal success probability as
\begin{subequations}
  \begin{align}
    P_{\cO}(\rho\xrightarrow{\gamma,\, \delta} \Psi_m,\ve) := \max &\ p\\
    \text{s.t.} &\ \Pi(\rho\ox \gamma) = p \ket{0}\bra{0}  \ox \sigma \ox \gamma_0 + (1-p) \ket{1}\bra{1}\ox \o \ox \gamma_1,\\
    &\ F(\sigma,\Psi_m) \geq 1-\ve,\, F(\gamma,\gamma_0) \geq 1-\delta,\, F(\gamma,\gamma_1) \geq 1-\delta,\\
    & \ \Pi \in \cO,\, 0 \leq p \leq 1.
  \end{align}
\end{subequations}
Using the same argument as Remark 1 and Remark 2, we can fix the output state $\sigma = \Psi_m^\ve$ and $\o = \1/m$. Then the maximal success probability can be written as 
\begin{subequations}\label{cata general}
  \begin{align}
    P_{\cO}(\rho\xrightarrow{\gamma,\, \delta} \Psi_m,\ve) = \max &\ p\\
    \text{s.t.} &\ \Pi(\rho\ox \gamma) = p \ket{0}\bra{0}  \ox \Psi_m^\ve \ox \gamma_0 + (1-p) \ket{1}\bra{1}\ox \1/m \ox \gamma_1,\\
    &\ F(\gamma,\gamma_0) \geq 1-\delta,\, F(\gamma,\gamma_1) \geq 1-\delta,\label{fidelity condition general}\\
    & \ \Pi \in \cO,\, 0 \leq p \leq 1.
  \end{align}
\end{subequations}
If the catalyst $\gamma$ we considered is the maximally coherent state $\Psi_k$, we can also fix the output catalyst $\gamma_0 = \gamma_1 = \Psi_k^\delta$. It gives
\begin{subequations}\label{SM DIO}
  \begin{align}
    P_{\cO}(\rho\xrightarrow{\Psi_k,\, \delta} \Psi_m,\ve) = \max &\ p\\
    \text{s.t.} &\ \Pi(\rho\ox \Psi_n) = p \ket{0}\bra{0}  \ox \Psi_m^\ve \ox \Psi_k^\delta + (1-p) \ket{1}\bra{1}\ox \1/m \ox \Psi_k^\delta,\\
    & \ \Pi \in \cO,\, 0 \leq p \leq 1.
  \end{align}
\end{subequations}
In the above optimization, the only variables we need to optimize over are $\Pi$ and $p$. It is thus clear that optimization~\eqref{SM DIO} is an SDP for DIO.

Consider the class of states $\rho=q \cdot v_1+ (1-q)\cdot v_2$
with
$\ket {v_1} = (\ket {00} -\ket {01} -\ket {10}+\ket {11})/2$, $\ket{v_2} = (2\ket {00} +6\ket {01} -3\ket {10}+\ket {11})/5\sqrt{2}$ and state parameter $q \in [0.1,0.5]$. We show the difference between the catalysis-assisted success probability $P_{\DIO}\big(\rho \xrightarrow{\Psi_2,\, \delta} \Psi_2,0.01\big)$ and the unassisted success probability $P_{\DIO}\big(\rho \!\to\! \Psi_2,0.01\big)$ in the following Fig.~\ref{Cata compare}. It shows that catalyst indeed enhances the maximal success probability even when requiring perfect catalyst to be returned ($\delta = 0$), while slight infidelity tolerance of the returning catalyst will boost the maximal success probability even more. On the right hand side of Fig.~\ref{Cata compare}, the enhancement ratio is given by $\big[P_{\DIO}\big(\rho \xrightarrow{\Psi_2,\, \delta} \Psi_2,0.01\big) - P_{\DIO}\big(\rho \!\to\! \Psi_2,0.01\big)\big]/P_{\DIO}\big(\rho \!\to\! \Psi_2,0.01\big)$. The enhancement ratio can achieve up to $12\%$ when $\delta=0$.

\begin{figure}[H]
\centering	
\begin{minipage}{0.45\textwidth}
	\centering
	\includegraphics[width=7cm]{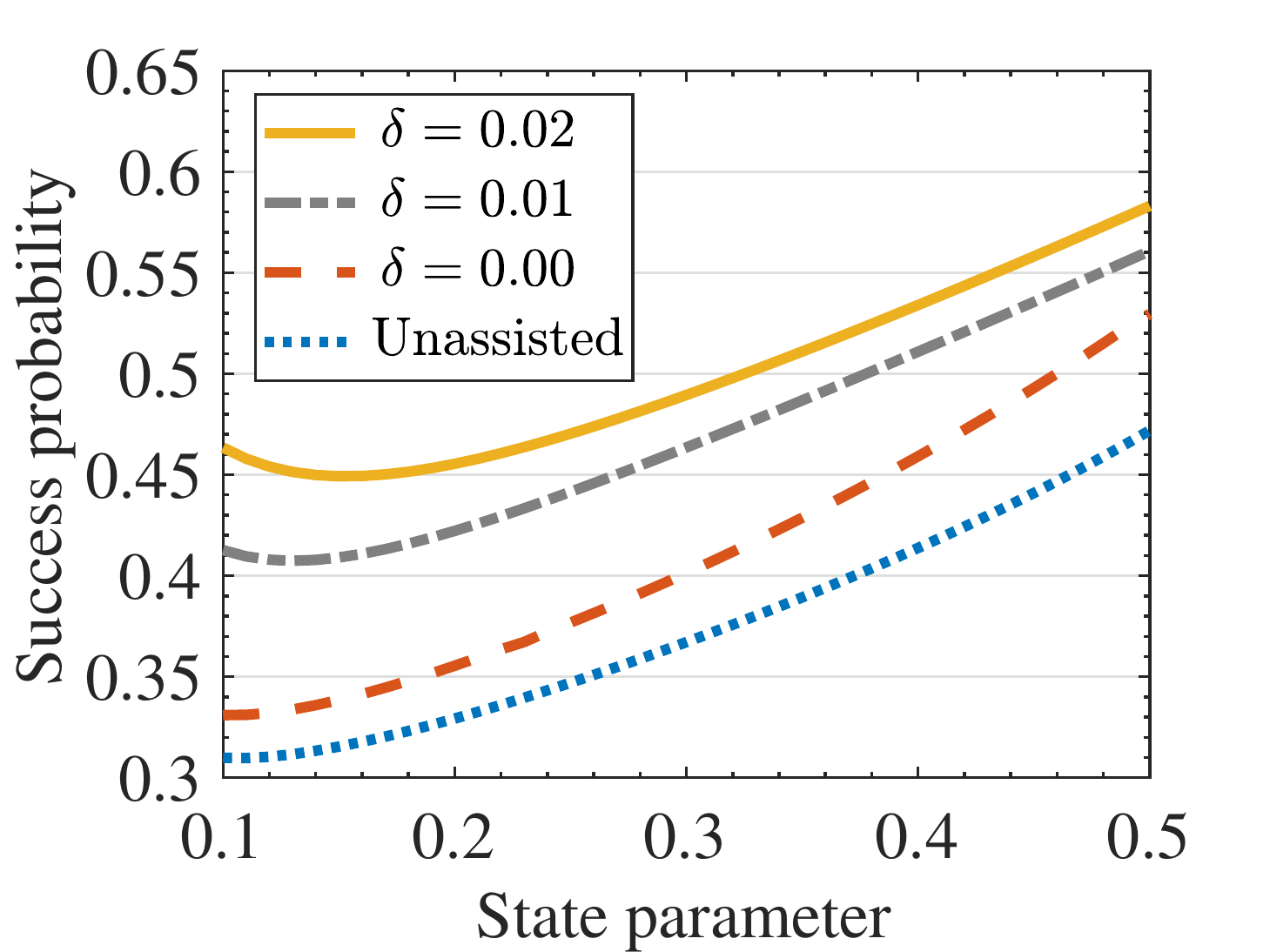}
\end{minipage}
\begin{minipage}{0.45\textwidth}
	\centering
	\includegraphics[width=7cm]{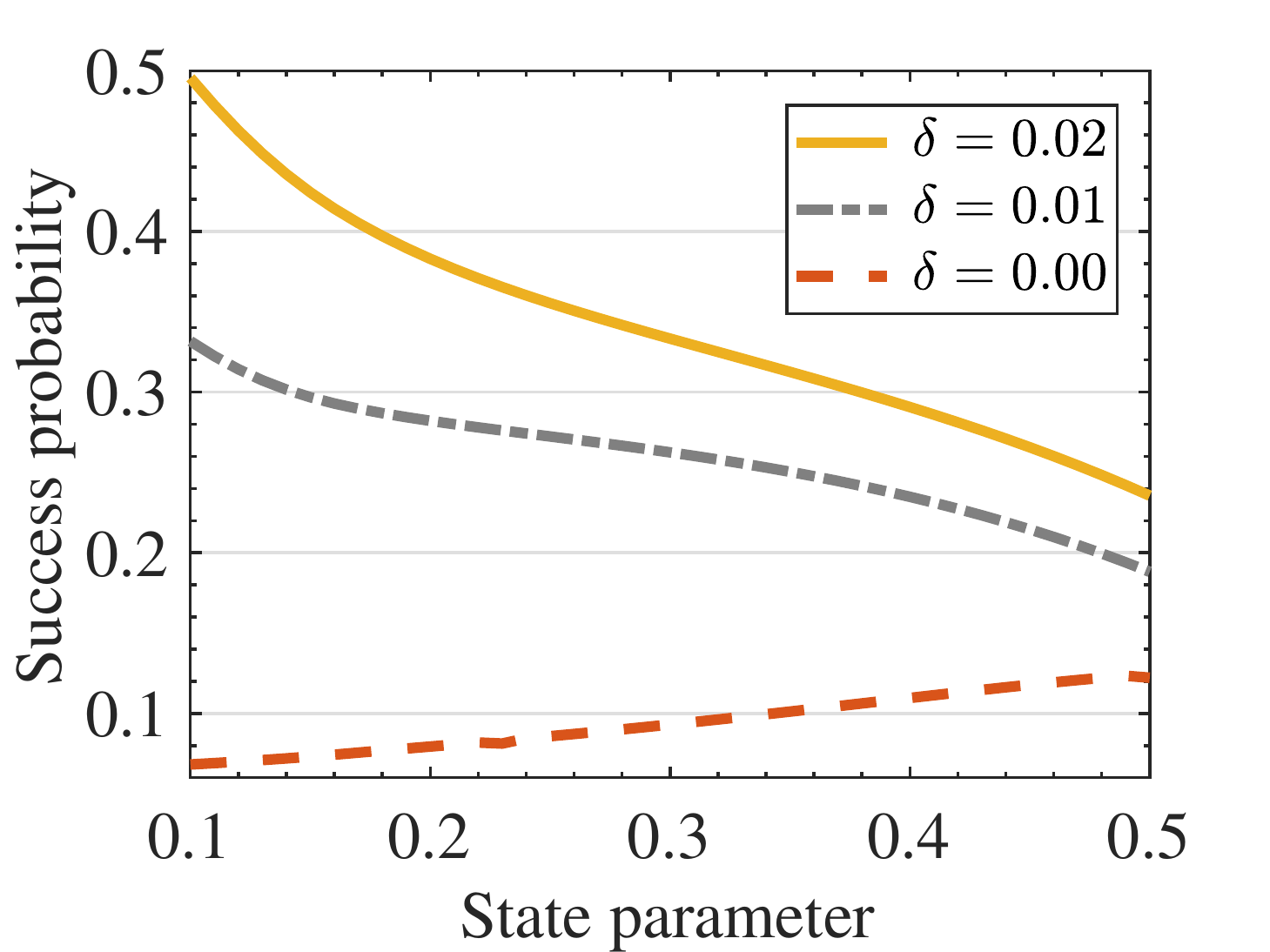}
\end{minipage}
\caption{\small Examples of catalyst-assisted probabilistic coherence distillation.}
\label{Cata compare}
\end{figure}

Similarly, we can also give another class of states showing the enhancement of catalyst. 
Consider the class of states $\rho = q\cdot u_1 +(1-q)\cdot u_2$ with $\ket{u_1} = (\ket{00}+\ket{01}+\ket{10}+\ket{11})/2$ and $\ket{u_2} = (3\ket{00}-2\ket{01}+\ket{10}+2\ket{11})/3\sqrt{2}$ and state parameter $q\in [0.2,0.7]$. We also compare $P_{\DIO}\big(\rho \xrightarrow{\Psi_2,\, \delta} \Psi_2,0.01\big)$ with $P_{\DIO}\big(\rho \!\to\! \Psi_2,0.01\big)$ and the enhancement is shown in Fig~\ref{Cata compare 2}.

\begin{figure}[H]
\centering  
\begin{minipage}{0.45\textwidth}
  \centering
  \includegraphics[width=9cm]{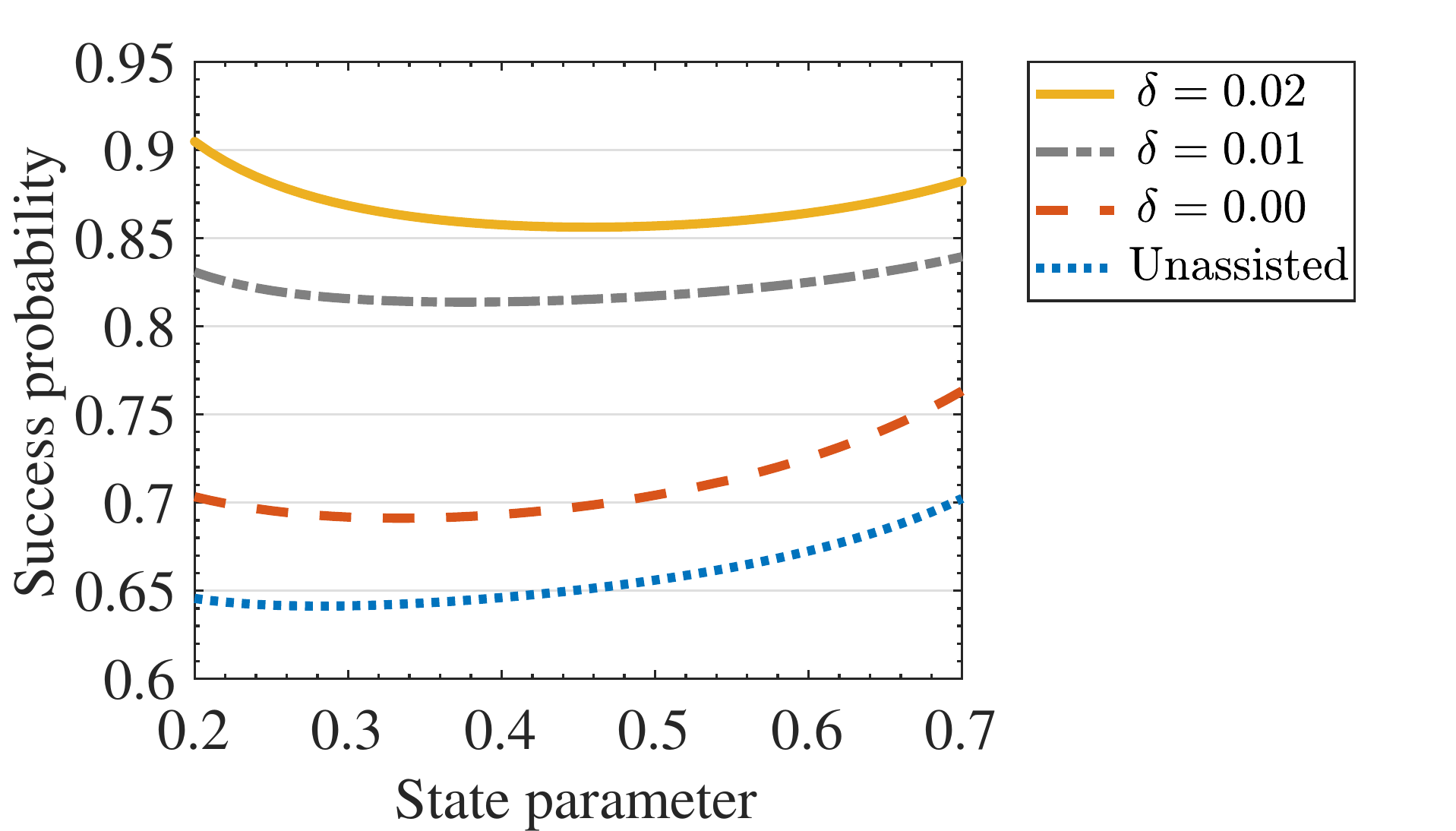}
\end{minipage}
\begin{minipage}{0.45\textwidth}
  \centering
  \includegraphics[width=7cm]{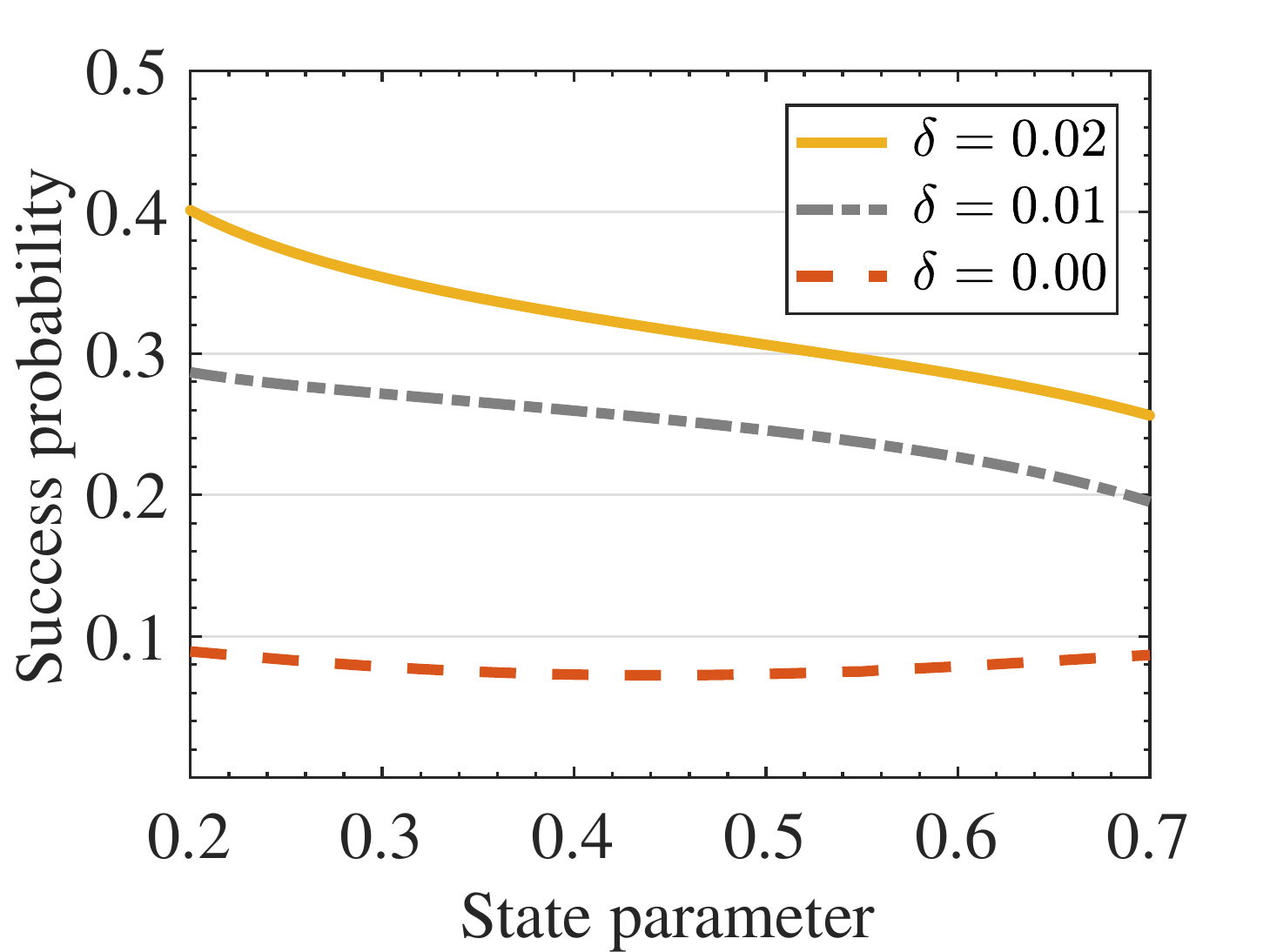}
\end{minipage}
\caption{\small Examples of catalyst-assisted probabilistic coherence distillation.}
\label{Cata compare 2}
\end{figure}

\begin{remark}
  We can denote $V=p\gamma_0$ and $W = (1-p)\gamma_1$ in the optimization \eqref{cata general}. Then for the case that catalyst $\gamma$ is a pure state, the fidelity conditions in \eqref{fidelity condition general} is equivalent to \eqref{fidelity condition pure state}. Combining with the semidefinite conditions for MIO/DIO, the maximal success probability under MIO/DIO can be written as SDPs:
 \begin{subequations}
  \begin{align}
    P_{\cO}(\rho\xrightarrow{\gamma,\, \delta} \Psi_m,\ve) = \max & \tr V\\
    \text{\rm s.t.} &\ \Pi(\rho\ox \gamma) = \ket{0}\bra{0}  \ox \Psi_m^\ve \ox V + \ket{1}\bra{1}\ox \1/m \ox W,\\
    & \tr \gamma V \geq (1-\delta) \tr V,\, \tr \gamma W \geq (1-\delta) \tr W \label{fidelity condition pure state}\\
    &\ \Pi \in \cO,\, 0 \leq p \leq 1.
  \end{align}
\end{subequations}
\end{remark}

\end{document}

%% file: pretex.tex
\usepackage{hyperref}
\hypersetup{colorlinks=true,citecolor=blue,linkcolor=blue,filecolor=blue,urlcolor=blue,breaklinks=true}
\usepackage{cleveref}
\usepackage{graphicx}
\usepackage{xcolor}

\definecolor{darkblue}{RGB}{0,76,156}
\definecolor{darkkblue}{RGB}{0,0,153}
\definecolor{blue2}{RGB}{102,178,255}

\usepackage{url}

\newtheorem{proposition}{Proposition}
\newtheorem{lemma}[proposition]{Lemma}

\newtheorem{theorem}[proposition]{Theorem}
\newtheorem{remark}{Remark}

\def\squareforqed{\hbox{\rlap{$\sqcap$}$\sqcup$}}
\def\qed{\ifmmode\squareforqed\else{\unskip\nobreak\hfil
\penalty50\hskip1em\null\nobreak\hfil\squareforqed
\parfillskip=0pt\finalhyphendemerits=0\endgraf}\fi}
\def\endenv{\ifmmode\;\else{\unskip\nobreak\hfil
\penalty50\hskip1em\null\nobreak\hfil\;
\parfillskip=0pt\finalhyphendemerits=0\endgraf}\fi}
\newenvironment{proof}{\noindent \textbf{{Proof.}~}}{\hfill $\blacksquare$}

\mathchardef\ordinarycolon\mathcode`\:
\mathcode`\:=\string"8000
\def\vcentcolon{\mathrel{\mathop\ordinarycolon}}
\begingroup \catcode`\:=\active
  \lowercase{\endgroup
  \let :\vcentcolon
  }

\makeatletter
\def\resetMathstrut@{%
    \setbox\z@\hbox{%
        \mathchardef\@tempa\mathcode`\[\relax
        \def\@tempb##1"##2##3{\the\textfont"##3\char"}%
        \expandafter\@tempb\meaning\@tempa \relax
    }%
    \ht\Mathstrutbox@\ht\z@ \dp\Mathstrutbox@\dp\z@}

\newcommand{\nc}{\newcommand}
\nc{\rnc}{\renewcommand}
\nc{\beg}{\begin{equation}}
\nc{\eeq}{{\end{equation}}}
\nc{\beqa}{\begin{eqnarray}}
\nc{\eeqa}{\end{eqnarray}}
\nc{\lbar}[1]{\overline{#1}}
\nc{\bra}[1]{\langle#1|}
\nc{\ket}[1]{|#1\rangle}
\nc{\ketbra}[2]{|#1\rangle\!\langle#2|}
\nc{\braket}[2]{\langle#1|#2\rangle}

\nc{\proj}[1]{| #1\rangle\!\langle #1 |}
\nc{\avg}[1]{\langle#1\rangle}
\nc{\Rank}{\operatorname{Rank}}
\nc{\smfrac}[2]{\mbox{$\frac{#1}{#2}$}}
\nc{\tr}{\operatorname{Tr}}
\nc{\ox}{\otimes}
\nc{\dg}{\dagger}
\nc{\dn}{\downarrow}
\nc{\cA}{{\cal A}}
\nc{\cB}{{\cal B}}
\nc{\cC}{{\cal C}}
\nc{\cD}{{\cal D}}
\nc{\cE}{{\cal E}}
\nc{\cF}{{\cal F}}
\nc{\cG}{{\cal G}}
\nc{\cH}{{\cal H}}
\nc{\cI}{{\cal I}}
\nc{\cJ}{{\cal J}}
\nc{\cK}{{\cal K}}
\nc{\cL}{{\cal L}}
\nc{\cM}{{\cal M}}
\nc{\cN}{{\cal N}}
\nc{\cO}{{\cal O}}
\nc{\cP}{{\cal P}}
\nc{\cQ}{{\cal Q}}
\nc{\cR}{{\cal R}}
\nc{\cS}{{\cal S}}
\nc{\cT}{{\cal T}}
\nc{\cV}{{\cal V}}
\nc{\cX}{{\cal X}}
\nc{\cY}{{\cal Y}}
\nc{\cZ}{{\cal Z}}
\nc{\cW}{{\cal W}}
\nc{\csupp}{{\operatorname{csupp}}}
\nc{\qsupp}{{\operatorname{qsupp}}}
\nc{\var}{{\operatorname{var}}}
\nc{\rar}{\rightarrow}
\nc{\lrar}{\longrightarrow}
\nc{\polylog}{{\operatorname{polylog}}}
\nc{\1}{{\mathds{1}}}
\nc{\wt}{{\operatorname{wt}}}
\nc{\av}[1]{{\left\langle {#1} \right\rangle}}
\nc{\supp}{{\operatorname{supp}}}

\nc{\dia}{{\diamondsuit }}

\newcommand{\Tr}{\text{Tr}\,}

\def\ve{\varepsilon}

\def\o{\omega}

\nc{\RR}{{{\mathbb R}}}
\nc{\CC}{{{\mathbb C}}}
\nc{\FF}{{{\mathbb F}}}
\nc{\NN}{{{\mathbb N}}}
\nc{\ZZ}{{{\mathbb Z}}}
\nc{\PP}{{{\mathbb P}}}
\nc{\QQ}{{{\mathbb Q}}}
\nc{\UU}{{{\mathbb U}}}
\nc{\EE}{{{\mathbb E}}}
\nc{\id}{{\operatorname{id}}}

\nc{\CHSH}{{\operatorname{CHSH}}}

\nc{\be}{\begin{equation}}
\nc{\ee}{{\end{equation}}}
\nc{\bea}{\begin{eqnarray}}
\nc{\eea}{\end{eqnarray}}
\nc{\<}{\langle}
\rnc{\>}{\rangle}
\nc{\Hom}[2]{\mbox{Hom}(\CC^{#1},\CC^{#2})}
\nc{\rU}{\mbox{U}}

\nc{\ob}[1]{#1}

\nc{\SEP}{{\text{SEP}}}
\nc{\NS}{{\text{NS}}}
\nc{\LOCC}{{\text{LOCC}}}
\nc{\PPT}{{\text{PPT}}}
\nc{\EXT}{{\text{EXT}}}
\nc{\Sym}{{\operatorname{Sym}}}

\DeclareMathOperator{\rank}{rank}

\nc{\HH}{\mathbb{H}}

\nc{\ERLO}{{E_{\text{r,LO}}}}
\nc{\ERLOCC}{{E_{\text{r,LOCC}}}}
\nc{\ERPPT}{{E_{\text{r,PPT}}}}
\nc{\ERLOCCinfty}{{E^{\infty}_{\text{r,LOCC}}}}
\nc{\Aram}{{\operatorname{\sf A}}}

\rnc{\bar}{\;\rule{0pt}{9.5pt}\right|\;}

\nc{\lset}{\left\{\left.}
\nc{\rset}{\right\}}

\nc{\lsetr}{\left\{}
\nc{\rsetr}{\right.\right\}}
\nc{\barr}{\left|\rule{0pt}{9.5pt}\;}

\DeclareMathOperator*{\argmin}{arg\,min}
\DeclareMathOperator{\Det}{Det}
\let\id\1
\newcommand{\cond}[1]{\text{(#1)}}

\nc{\norm}[2]{\left\lVert#1\right\rVert_{#2\!}}
\nc{\lnorm}[2]{\left\lVert#1\right\rVert_{\ell_{#2}}}

\usepackage{txfonts,newtxmath}